\documentclass[11pt]{article}

\usepackage[letterpaper,top=3cm,bottom=3cm,left=3cm,right=3cm,marginparwidth=1.75cm]{geometry}

\usepackage{graphicx} % Required for inserting images
\usepackage[colorlinks=true, allcolors=blue]{hyperref}
\usepackage{alltt}
\usepackage{amsfonts}
\usepackage{amsmath, amsthm, bm}
\usepackage{amssymb}
\usepackage{tabularx}
\usepackage{subcaption}
\usepackage{prodint}
\usepackage{aligned-overset}
\usepackage[bb=dsserif]{mathalpha}
\usepackage{setspace}

\usepackage[
    style=apa,
  ]{biblatex}
\usepackage{enumitem}
\newtheorem{assumption}{Assumption}

\newtheoremstyle{boldremark}
    {\dimexpr\topsep/2\relax} % space above
    {\dimexpr\topsep/2\relax} % space below
    {}          % body font
    {}          % indent amount
    {\bfseries} % theorem head font
    {.}         % punctuation after theorem head
    {.5em}      % space after theorem head
    {}          % theorem hed spec. (empty = "normal")

\DeclareMathOperator{\E}{\mathrm{E}}
\DeclareMathOperator{\var}{\mathrm{var}}
\DeclareMathOperator{\cov}{\mathrm{cov}}
\newcommand\norm[1]{\left\lVert#1\right\rVert}
\newcommand\dd{\mathop{}\!\mathrm{d}}
\newcommand\indep{\protect\mathpalette{\protect\independenT}{\perp}}
\def\independenT#1#2{\mathrel{\rlap{$#1#2$}\mkern2mu{#1#2}}}
\newtheorem{theorem}{Theorem}
\newtheorem{lemma}{Lemma}
\newtheorem{corollary}{Corollary}
\newtheorem{alemma}{Lemma}[section]
\theoremstyle{boldremark}

\newcommand{\nothere}[1]{}

\title{Causal effect on the number of life years lost due to a specific event: Average treatment effect and variable importance}
\author{Simon Christoffer Ziersen $\&$ Torben Martinussen \\ \\ \small \textit{Section of Biostatistics, University of Copenhagen}}
\date{}

\begin{document}

\maketitle

\begin{abstract}
Competing risk is a common phenomenon when dealing with time-to-event outcomes in biostatistical applications. An attractive estimand  in this setting is the "number of life-years lost due to a specific cause of death", \textcite{andersen}. It provides a direct interpretation on the time-scale on which the data is observed. In this paper, we introduce the causal effect on the number of life years lost due to a specific event, and we give assumptions under which the average treatment effect (ATE) and the conditional average treatment effect (CATE) are identified from the observed data. Semiparametric estimators for ATE and a partially linear projection of CATE, serving as a variable importance measure, are proposed. These estimators leverage machine learning for nuisance parameters and are model-agnostic, asymptotically normal, and efficient. We give conditions under which the estimators are asymptotically normal, and their performance are investigated in a simulation study. Lastly, the methods are implemented in a study concerning the response to different antidepressants using data from the Danish national registers.   
\end{abstract}


\section{Introduction} 

Time-to-event outcomes are very common in the biomedical sciences, and often there is the additional complication of competing events. For example, \textcite{kessing} were interested in the time to non-response to different antidepressants for patients with a major depressive disorder. The event of interest was defined as a switch to or add-on of another antidepressant or antipsychotic medicine or readmission to a psychiatric hospital with a major depressive disorder, while admission to a psychiatric hospital with a higher order psychiatric diagnosis (bipolar disorder, schizophrenia or organic mental disorder) or death were competing risks. One purpose of that study was to
compare a potential causal effect of treatment with the antidepressant Setraline versus treatment with another antidepressant, Escitalopram, and to determine whether such an effect was dependent on individual characteristics.
With competing risk data its standard to target the average treatment effect (ATE) based on the 
cumulative incidence functions associated with the competing events. Estimation of the ATE in this setting has been described by  \textcite{ozenne} under working Cox models and logistic regression for the nuisance parameters, whereas \textcite{heleneFrank} and \textcite{heleneMark} provided estimators of the ATE under data-adaptive nuisance estimators. 

\textcite{andersen} suggested an alternative estimand that has an attractive interpretation.  Within a given time frame $[0,t]$ it aims, in the classical setting of competing causes of death, to decompose the years of life lost (in the considered time window) due to the specific causes of death. In our working example, where the first two years after the initiation of treatment were of most interest, it corresponds to healthy days lost before two years after the   treatment initiation
due to the primary cause and the competing causes.
In this paper, we introduce the ATE based on this quantity - number of life years lost due to a specific event. A further benefit of this estimand is that it has an interpretation directly on the timescale inherent in the data, and it may further capture possible "early" effects of treatment. We provide an estimator of the ATE based on semiparametric efficiency theory, allowing for the use of machine learning methods for nuisance parameter estimation without relying on model-specifications. We further give assumptions on the nuisance estimators under which the proposed ATE estimator is asymptotically normal and nonparametric locally efficient. 

\textcite{andersenParner} provide a general procedure for estimating the ATE with censored data, where the ATE based on the number of life years lost due to a specific event serves as an example. Their estimator is based on pseudo-observations and its asymptotic properties in the presence of covariate dependent censoring are derived in \textcite{overgaard2019}. Specifically, their estimator relies on a correctly specified parametric model for the conditional average treatment effect along with a correctly specified model for the censoring distribution. In contrast, our estimator does not rely on such parametric model specifications, and it allows for estimation of the involved nuisance parameters based on machine learning, while still providing inference for the ATE estimate. Importantly, we also extend the treatment effect variable importance measure given in \textcite{zm} as a best partially linear projection of the conditional average treatment effect. The projection parameter provides an attractive  measure of potential treatment effect heterogeneity through a given covariate - this is of great interest in many real applications as the one we present in Section \ref{sec:application} on evaluating antidepressant treatment effects as mentioned above. We develop  an estimator of this estimand that is also semiparametrically efficient and allows for machine learning methods to estimate nuisance parameters. The estimator admits an asymptotic normal distribution which defines a test of treatment effect heterogeneity of a given covariate. 

In Section 2, we state the notation and setup used in the paper, and in Section 3 we define two target parameters. Section 4 gives the efficient influence functions for the two target parameters and utilize these to construct cross-fitted one-step estimators. The asymptotic distributions of the estimators are then proved under high level assumptions on the nuisance parameter estimates. In Section 5, the finite sample performance of the proposed estimators are investigated in a simulation study and in Section 6 we apply the estimators to a study on treatment response to different antidepressants based on data from the Danish national registers (\cite{kessing}). Section 7 concludes the paper with some final remarks.

\section{Notation and setup}
We consider a time-to-event setting with competing risks. Let $T$ be the time to event and $\Delta \in \{1,2\}$ the event indicator for two competing events. Let $X$ be a $d$-dimensional vector of covariates, and let $A$ denote the baseline treatment indicator. We enforce censoring through a censoring time $C$, such that the observed event time is $\tilde{T} = T \wedge C$ and the observed event indicator is $\tilde{\Delta} = \mathbb{1}(C\geq T)\Delta$. Our observed data, $\mathcal{O}$, is given by $n$ i.i.d. copies of $O = (\tilde{T}, \tilde{\Delta}, A, X)\sim P_0$, where $P_0 \in \mathcal{M}$, with $\mathcal{M}$ being a nonparametric model. 

We introduce the conditional cause-specific hazard functions, $\lambda_{0, j}(t\mid a, x)$, for the $j$'th cause, $j=1,2$, and let $\lambda_{0,c}(t\mid a,x)$ denote the censoring hazard function. We let $\Lambda_{0, j}(t\mid a, x) = \int_0^t \lambda_{0, j}(s\mid a, x)\dd s$ and $\Lambda_{0, C}(t\mid a, x) = \int_0^t \lambda_{0, C}(s\mid a, x) \dd s$ denote the corresponding cumulative hazard functions. We denote $S(t\mid a, x) = \exp\{-\Lambda_1(t\mid a,x) - \Lambda_2(t\mid a,x)\}$ the survival function, and $\pi_0(a\mid x) = P_0(A = a\mid X=x)$ the conditional distribution of $A$ given $X$. Let $N_j(t) = \mathbb{1}(\tilde{T} \leq t, \tilde{\Delta} = j)$ denote the observed counting process for the $j$'th event and let $M_j(t\mid a,x)$ denote the corresponding martingale conditional on $A=a$ and $X=x$ such that $M_j(\dd t\mid a,x) = N_j(\dd t) - \mathbb{1}(\tilde{T}\geq t)\Lambda_{0,j}(\dd t\mid a,x)$. Furthermore, we introduce the cause-specific event times $T_j$, $j = 1,2$, and let $T_j^a$, $a = 0,1$, denote the counterfactual time corresponding to the $j$'th cause.

We use the notation $Pf = \int f \dd P$ and $\mathbb{P}_n f = \sum_{i=1}^n f(O_i)$, and $E_0 \{ f(O) \} = \int f \dd P_0$ is the expectation of $f(O)$ under the true data generating distribution. Throughout, all expectations, $P\hat{f}$, considers the function $\hat{f}$ fixed, even when it is estimated from the data, unless otherwise specified. Finally, $\norm{\cdot}$ denotes the $L_2(P)$-norm, such that $\norm{f} = \left( \int f^2 \dd P \right)^{1/2}$.

\section{Causal estimand and nuisance parameters}\label{sec:estimands}
Inspired by \textcite{andersen}, we introduce 
$$
L_0(0,t^*|a,x)=t^*-\int_0^{t^*}S(u|a,x)du
$$
for a given time-horizon $[0, t^*]$, which can be interpreted as the expected number of years lost before time $t^*$ in strata $(a,x)$. As \textcite{andersen} shows, this quantity can be decomposed naturally into 
$$
L_0(0,t^*|a,x)=L_1(0,t^*|a,x)+L_2(0,t^*|a,x)
$$
where 
$$
L_j(0,t^*|a,x)=\int_0^{t^*}F_j(u|a,x)du,\quad j=1,2,
$$
can be interpreted as number of years lost "due to cause $j$" (\cite{andersen}), with $F_j$ being the $j$th cumulative incidence function given $A=a$ and $X=x$, i.e. $F_j(t\mid a, x) = \int_0^t S(s\mid a, x)\dd \Lambda_j(s\mid a, x)$.
To introduce the counterfactual number of life years lost due to a specific event, we first remark on an observation given in \textcite{andersen}. The random variable $T_j$ is improper because $P(T_j = \infty) > 0$, but the random variable $T_j\wedge t^*$ is proper with expectation given by
$$
E\{T_j\wedge t^*\} = t^* - \int_0^{t^*} F_j(s) \dd s,
$$
and hence 
$$
E(T_j\wedge t^* \mid a, x ) = t^* - \int_0^{t^*} F_j(s\mid a, x) \dd s.
$$
We now introduce the counterfactual $Y^a_j(t^*) = t^* - T_j^a \wedge t^*$ for $a=0,1$, which is the number of life-years lost due to event $j$ before time $t^*$ under treatment $a$. We define the $j$'th-specific ATE as 
$$
E_0\{Y_j^1(t^*) - Y_j^0(t^*)\}
$$
and the CATE is
$$
E_0(Y_j^1(t^*) - Y_j^0(t^*)\mid X=x).
$$
In order to identify the ATE and CATE from the observed data we need the following assumptions:
\begin{NoHyper}
\begin{assumption}[Identification]\label{ass:ident}
    \phantom{something}
    \begin{enumerate}[label=\ref{ass:ident}\arabic*]
        \item \label{ass:cons}
        (Consistency) $Y_j(t^*) = t^* - T_j\wedge t^* = AY_j^1(t^*) + (1-A)Y_j^0(t^*)$ conditional on $A$. 
        \item \label{ass:exch}
        (Exchangeability) $Y^a_j(t^*) \indep A \mid X, \ a=0,1.$ 
        \item \label{ass:pos}
        (Positivity) $\pi(a\mid x)P_0(C>t| A=a, X=x)P_0(T > t| A=a,X=x) > \eta > 0, \ \forall (t,x) \in [0,t^*]\times \mathcal{X}, \ a=0,1$. 
        \item \label{ass:indepCens}
        (Independent censoring) $T\indep C \mid A, X$.
    \end{enumerate}
\end{assumption}
\end{NoHyper}

Define
$$
\tau_j(x;t^*) \equiv L_j(0,t^*|1,x)- L_j(0,t^*|0,x),\quad j=1,2.
$$
In Appendix \ref{app:identification}, we show that the CATE function is identified in the observed data because
$$
\tau_j(x;t^*) = E_0(Y_j^1(t^*) - Y_j^0(t^*)\mid X=x),
$$
and the average treatment effect as
$$
E_0\{\tau_j(X;t^*)\} = E_0\{Y_j^1(t^*) - Y_j^0(t^*)\}
$$
under assumption \ref{ass:ident}. Going forward we drop the dependence of $t^*$ and write $\tau_j(x) = \tau_j(x;t^*)$ to ease notation. Based on the identification results, we define two target parameters as mappings from the model $\mathcal{M}$ on the observed data to the reals. The first parameter is the $j$-specific average treatment effect, defined as the mapping $\psi_j : \mathcal{M}\rightarrow \mathbb{R}$, where
$$
\psi_j(P) = E\{L_j(0,t^*|1,X)- L_j(0,t^*|0,X) \}.
$$
The second parameter is defined as a variable importance measure of the $l$'th covariate on $\tau_j(x)$ based on the best partially linear projection given in \textcite{zm}. It is defined as the mapping $\Omega_j^l: \mathcal{M}\rightarrow \mathbb{R}$ with
$$
\Omega_j^l(P) = \frac{E\{\cov(X_l, \tau_j(X) \mid X_{-l})\}}{E\{\var(X_l\mid X_{-l})\}},
$$
where $X_{-l}$ denotes the covariates indexed by $\{1, \ldots, d\} \setminus \{l\}$. The parameter $\Omega_j^l$ can be viewed as a weighted average of the conditional covariance of the CATE and the covariate $X_l$ given the rest of the covariates. The parameter is dependent on the scale of the covariate in question, and when assessing the importance of different covariates it is not the estimate of the parameter that determines the ranking of variable importance, but rather the p-value associated with test $H:\Omega_j^l = 0$, since $\Omega_j^l$ is zero if there is no heterogeneity explained by $X_l$. The parameter can be derived as the least-squares projection of the CATE onto the partially linear model. 
%For more details and discussion, see \textcite{zm}.

\section{Estimation and inference}\label{sec:est}
We base the estimation of the two target parameters, $\psi_j(P)$ and $\Omega_j^l(P)$, on semiparametric efficiency theory (\cite{bickel}, \cite{Vaart} Ch. 25, \cite{vanRose}). 

For a general target parameter $\psi$, an estimator $\hat{\psi}$ is said to be asymptotically linear if it can be written on the form $\hat{\psi} - \psi = \mathbb{P}_n\mathbb{IF} + o_p(n^{-1/2})$ with $P\mathbb{IF}=0$. The function $\mathbb{IF}$ is called the influence function of the estimator $\hat{\psi}$, and it characterizes the asymptotic distribution of the estimator. This can be seen by applying the central limit theorem together with Slutsky's lemma, from which $\sqrt{n}(\hat{\psi}-\psi) \overset{D}{\rightarrow} \mathcal{N}(0, P\mathbb{IF}^2)$. If the target parameter is differentiable at $P$ as a map $\psi:\mathcal{M}\rightarrow \mathbb{R}$, there exist a unique function, say $\tilde{\psi}$, associated to the pathwise derivative of the parameter, which characterizes the information bound of any regular estimator. The function $\tilde{\psi}$ is called the efficient influence function (EIF), and any estimator is regular and asymptotically efficient, if it is asymptotically linear with influence function given by the EIF (\cite{Vaart} Ch. 25.3).

Since the EIF is uniquely determined by the target parameter and the model, it can be calculated without reference to any estimator. Once it is known, several techniques exist for constructing estimators that are asymptotically linear with the EIF as their influence function (\cite{vanRose}, \cite{cher}, \cite{kennedydouble}, \cite{hinesDemystifying}). We focus on the so-called \textit{one-step estimator}, which is defined as 
$$
\hat{\psi}^{OS} = \psi(\hat{P}) + \mathbb{P}_n\tilde{\psi}(\cdot;\hat{P}),
$$
where $\hat{P}$ is obtained from some (possibly) data-adaptive estimators. In order to show that the one-step estimator is asymptotically linear, one typically relies on a certain decomposition involving an empirical process term and a remainder term, which are both required to be $o_p(n^{-1/2})$. The former can be obtained if $\hat{P}$ is assumed to belong to a Donsker class, but this requirement has been shown to be too restrictive for some data-adaptive estimators, and a certain type of sample splitting, termed \textit{cross-fitting}, has to be applied to the one-step estimator in order to obtain the required convergence rate (\cite{cher}, \cite{kennedydouble}).
\subsection{Average treatment effect}\label{sec:estATE}
To derive an estimator for the ATE we first derive its EIF. We define the nuisance parameter $\nu = (\Lambda_1, \Lambda_2, \Lambda_c, \pi)$. 
The EIF is then given in the following lemma.
\begin{lemma} \label{lem:gatpsi}
The efficient influence function of $\psi_j(P)$ is given by
$$\tilde{\psi}_{\psi_j}(O;\nu) = \varphi_j(\nu)(O) - \psi_j(P),$$
where $\varphi_j(\nu)$ is a real-valued function defined on the sample space of $O$ at a given value of $\nu$ with 
\begin{align}
    \varphi_j(\nu)(O) = \tau_j(X) + \left(\frac{\mathbb{1}(A=1)}{\pi(1\mid X)} - \frac{\mathbb{1}(A=0)}{\pi(0\mid X)}\right)\biggl\{\sum_{i=1,2} \int_0^{t^*} \frac{H_{ij}(s,t^*\mid A, X)}{S_C(s\mid A, X)} \dd M_i(s\mid A, X)\biggr \} 
\end{align}
where
\begin{align}
   H_{ij}(s,t\mid a,x) = \int_s^t \mathbb{1}(i=j) + \frac{F_j(s\mid a, x) - F_j(u\mid a, x)}{S(s\mid a, x)} \dd u. 
\end{align}
\end{lemma}
\begin{proof}
    See Appendix \ref{app:eif}.
\end{proof}

The function $\varphi_j(\nu)$ is the uncentered EIF of the ATE and will also appear in the development of an estimator for the best partially linear projection in Section \ref{est:BPLIP}. As the EIF of the ATE is linear in the target parameter, the one-step estimator for $\psi_j$ reduces to
$$
\hat{\psi}_j^{OS} = \mathbb{P}_n \varphi_j(\hat{\nu}).
$$
As noted earlier, the one-step estimator may fail to be asymptotically linear when using data-adaptive estimators for $\hat{\nu}$ and we further have to use a cross-fitted version of the estimator in order to obtain the desired asymptotic properties. 

To define the sample splitting involved in constructing the cross-fitted estimator, let $\boldsymbol{i} = (i_1, i_2, \ldots, i_n)$ be an index vector drawn from an $n$-dimensional multinomial distribution with $K$ events with probability $p_k=\frac{1}{K}$, $k=1,\ldots, K$, for the $k$'th event. Define the index sets $\mathcal{T}_k = \{j : i_j = k\}$ for $k=1\ldots K$ such that $\{1, \ldots, n\} = \dot{\cup}_{k=1}^K \mathcal{T}_k$, where $\dot{\cup}$ denotes the disjoint union. Corresponding to the index sets, we define $K$ disjoint data splits by $\mathcal{V}_k = \{O_j: i_j = k\}$ such that $\mathcal{O} = \dot{\cup}_{k=1}^K \mathcal{V}_k$, and we define $K$ leave-out data splits by $\mathcal{V}_{-k} = \dot{\cup}_{j \neq k} \mathcal{V}_j$.

To construct a cross-fitted one-step estimator of $\psi_j$ based on the EIF, let $\hat{\nu} = (\hat{\Lambda}_1, \hat{\Lambda}_2, \hat{\Lambda}_c, \hat{\pi}) $ denote the estimated nuisance parameter, and let $\hat{\nu}_{-k}$ be the estimated nuisance parameter based on data in the $k$'th leave out sample, $\mathcal{V}_{-k}$, and let $\mathbb{P}_n^k$ be the empirical measure of $O \in \mathcal{V}_k$. We define the K-fold cross-fitted estimator of $\psi_j$ as
\begin{align}
    \hat{\psi}_j^{CF} = \sum_{k=1}^K \frac{n_k}{n} \mathbb{P}_n^k \varphi(\hat{\nu}_{-k}) = \frac{1}{n} \sum_{k=1}^K \sum_{i \in \mathcal{T}_k} \varphi_j(\hat{\nu}_{-k})(O_i),
\end{align}
where $n_k$ is the number of observations in the $k$'th data split. The construction of the cross-fitted estimator detailed here is quite general, see \textcite{kennedydouble} for a nice discussion and comparison to the one-step estimator without cross-fitting.

In order to derive asymptotic results for $\hat{\psi}_j^{CF}$ we need a set of assumptions on the nuisance estimators. The assumptions below are stated for a general nuisance estimator $\hat{\nu}$, but in applications to cross-fitted estimators, they are assumed to hold for each leave-out sample $\mathcal{V}_{-k}$. We will be explicit about this when stating results on the obtained estimators, but it is left out of assumption \ref{ass:nuis} for notational convenience.

\begin{NoHyper}
\begin{assumption}\label{ass:nuis}
     The nuisance estimator $\hat{\nu}$ satisfy the following conditions
    \begin{enumerate}[label=\ref{ass:nuis}\arabic*]
        \item \label{ass:nuisPos} There exist a real-valued parameter $\eta >0$ such that $\hat{S}(t\mid a,x) > \eta$, $S(t\mid a,x) > \eta$, $\hat{S}_C(t\mid a,x) > \eta$, $S_C(t\mid a,x) > \eta$, $\hat{\pi}(a \mid x) > \eta$, $\pi(a \mid x) > \eta$ for all $(t,a,x) \in [0,t^*] \times \{0,1\} \times \mathcal{X}$.
        \item \label{ass:nuisCons} For $a=0,1$, it holds that
        \begin{align*}
            \E_0\left[\sup_{s\leq t^*} \left| \hat{\Lambda}_1(s\mid a, X) - \Lambda_{1,0}(s\mid a, X)\right|\right]^2 &= o_p(1) \\
            \E_0\left[\sup_{s\leq t^*} \left| \hat{\Lambda}_2(s\mid a, X) - \Lambda_{2,0}(s\mid a, X)\right|\right]^2 &= o_p(1) \\
            \E_0\left[\sup_{s\leq t^*} \left| \hat{\Lambda}_c(s\mid a, X) - \Lambda_{c,0}(s\mid a, X)\right|\right]^2 &= o_p(1) \\
            \E_0\left[\hat{\pi}(a\mid X) - \pi_{0}(a\mid X)\right]^2 &= o_p(1)
        \end{align*} 
        \item \label{ass:nuisDouble} For $a=0,1$, it holds that
        \begin{align*}
        &\E_0\left\{\sum_{i=1,2} \int_0^{t^*} S(s\mid a, X) \hat{H}_{ij}(s, t^*\mid a, X) \right. \\
        & \left. \phantom{\E\sum_{i=1,2} }\times \left(1 - \frac{\pi(a\mid X)S_C(s\mid a, X)}{\hat{\pi}(a\mid X)\hat{S}_C(s\mid a, X)} \right)\dd \left[ 
        \hat{\Lambda}_i(s\mid a, X) - \Lambda_i(s\mid a, X) \right] \right\} = o_p(n^{-1/2}).
        \end{align*}
    \end{enumerate}
\end{assumption}
\end{NoHyper}

Assumption \ref{ass:nuisPos} is the usual positivity assumption found through out the causal inference literature (for examples with censored data, see e.g. \cite{westling}, \cite{heleneFrank}, \cite{heleneMark}). Whereas assumption \ref{ass:pos} relates to the true data generating mechanism, assumption \ref{ass:nuisPos} extends to the estimators as well. We note that employing cross-fitting in relatively small sample sizes can sometimes lead to practical positivity-violations, when a "rare" covariate lies in $\mathcal{V}_k$ but not in $\mathcal{V}_{-k}$. For \ref{ass:nuisCons} to hold, all nuisance estimators must be consistent, and for the hazard estimators this amounts to uniform consistency. This requirement suggest the use of flexible learners for nuisance estimation. Assumption \ref{ass:nuisDouble} reflects the dobule robustness that is common for one-step estimators. In Lemma \ref{lem:remate_a} in the Appendix it is shown that the so-called remainder term, coming from the aforementioned decomposition of the estimator, takes this form. It is sometimes referred to as a second order remainder term, when it can be shown to hold if each of the nuisance estimators converge on $n^{-1/4}$-rate in $L_2(P)$-norm. If one considers cumulative hazard estimators that are absolute continuous, the result can be obtained from a simple application of the Cauchy-Schwarz inequality (for examples involving Highly Adaptive Lasso, see \cite{munch}, \cite{heleneFrank}, \cite{heleneThomas}), but this exclude many commonly used cumulative hazard estimators such as any Breslow-type estimators. We expect nonetheless that the double robustness is obtained for most reasonable estimators, and in the later simulation studies, this will be exemplified by the use of random survival forests (\cite{ishwaran}).

Next follows our main result for the cross-fitted ATE estimator:

\begin{theorem}\label{thm:anATE} Assume that the nuisance estimators $\hat{\nu}_{-k}$, $k=1\ldots K$ follow assumption \ref{ass:nuis} for each $k$. Then the cross-fitted estimator is asymptotically linear with influence function given by $\tilde{\psi}_{\psi_j}$ and hence
$$
\sqrt{n}\left(\hat{\psi}_j^{CF} - \psi_j \right) \overset{d}{\rightarrow} \mathcal{N}(0, P\tilde{\psi}_{\psi_j}(\cdot, \nu_0)^2).
$$
\end{theorem}
\begin{proof}
    See Appendix \ref{app:asympLin}.
\end{proof}
In practise, the variance of the influence function is estimated by the cross-fitted estimator:

$$
\hat{\sigma}_{\psi_j}^{2,CF} = \sum_{k=1}^K \frac{n_k}{n} \mathbb{P}_n^k\left(\varphi(\hat{\nu}_{-k}) - \hat{\psi}_j^{CF}\right)^2,
$$
and the standard error of $\hat{\psi}_j^{CF}$ is given by $\sqrt{\hat{\sigma}_{\psi_j}^{2,CF}/n}$. We note that $\hat{\sigma}_{\psi_j}^{2,CF}$ is a type of plug-in estimator, and contrary to the one-step estimator, it is not debiased via its influence function. Hence, even though the estimator is consistent (by Lemma 1 in \cite{zm}), for consistent nuisance estimators, it is not generally asymptotically linear for data-adaptive nuisance estimators, which may result in biased standard error estimates in finite samples. In the simulation studies in Section \ref{sec:sim}, this will be explored by contrasting the use of (semi)parametric and data-adaptive nuisance estimators, with the latter obtained from random forests.    

\subsection{Best partially linear projection}\label{est:BPLIP}
Estimation of $\Omega_j^l(P)$ follows the same overall strategy, i.e., construct an asymptotically linear estimator using semiparametric theory. The difference now is that the target parameter is a ratio of two parameters which can both be written as a map from $\mathcal{M}$ to the reals. If we construct asymptotically linear estimators for both parameters in the ratio, separately, then the ratio of the estimators will also be asymptotically linear. Furthermore, if each estimator in the ratio has their respective EIF as their influence function, the ratio of the estimators will have its EIF as its influence, by the functional delta method (\cite{Vaart} Ch. 25.7). In the following we will extend the approach of \textcite{zm} to the CATE function defined by the number of life years lost, $\tau_j(x)$, for a given time-horizon $[0,t^*]$, where each of the parameters in the ratio of $\Omega_j^l$ is estimated separately. We start by calculating the EIF for the relavant parameters.

\begin{lemma}\label{lem:gatvim}
Define the mappings $\Gamma_j^l:\mathcal{M}\rightarrow \mathbb{R}$ and $\chi_j^l:\mathcal{M}\rightarrow \mathbb{R}$ as 
$$
\Gamma_j^l(P) = E\{\cov(X_l, \tau_j(X) \mid X_{-l})\}
$$
and 
$$
\chi^l(P) = E\{\var(X_l\mid X_{-l})\}
$$
such that $\Omega_j^l(P) = \frac{\Gamma_j^l(P)}{\chi^l(P)}$. The efficient influence functions of $\Gamma_j^l(P)$, $\chi^l(P)$ and $\Omega_j^l(P)$, respectively, are given by
\begin{align}
    \tilde{\psi}_{\Gamma_j^l}(O;P) &= [\varphi_j(\nu)(O) - \E(\tau_j(X)\mid X_{-l})][X_l - \E(X_l\mid X_{-l})] - \Gamma_j^l(P), \\
\tilde{\psi}_{\chi^l}(O;P) &= [X_l - \E(X_l\mid X_{-l})]^2 - \chi^l(P), \\
\tilde{\psi}_{\Omega_j^l}(O;P) &= \frac{1}{\chi^l(P)}\left(\tilde{\psi}_{\Gamma_j^l}(O;P) - \Omega_j^l(P) \tilde{\psi}_{\chi^l}(O;P) \right).
\end{align}
\end{lemma}
\begin{proof}
    See Appendix \ref{app:eif}
\end{proof}
The EIF's above depend explicitly on the conditional distribution of $X_l$ given $X_{-l}$ through $E(X_l\mid X_{-l})$ and $E(\tau_j(X)\mid X_{-l})$, so to express them as mappings of the nuisance parameter, we extend the notion of $\nu$. Let $\tau_j^l(x_{-l}) = E(\tau_j(X)\mid X_{-l} = x_{-l})$ and $E^l(x_{-l}) = E(X_l\mid X_{-l} = x_{-l})$, and define $\nu_l^1 = E^l$ and $\nu_l^2 = (\Lambda_1, \Lambda_2, \Lambda_c, \pi, \tau_j^l, E^l)$. We define the uncentered EIF's corresponding to the EIF's $\tilde{\psi}_{\Gamma_j^l}$ and $\tilde{\psi}_{\chi_j^l}$ as
\begin{align}
    \phi_{\Gamma_j^l}(O; \nu_l^2) &= [\varphi_j(\nu)(O) - \tau_j^l(X_{-l})][X_l - E^l(X_{-l})] \\
    \phi_{\chi^l}(O; \nu_l^1) &= [X_l - E^l(X_{-l})]^2.
\end{align}
The construction of the estimators for $\Gamma_j^l$ and $\chi^l$ now follow similar to the procedure in the ATE setting. The estimation of $\chi^l$ is given in \textcite{zm}, but is included here as well for completeness. Let $\hat{\nu}^1_l$ and $\hat{\nu}^2_l$ denote the estimated nuisance parameters. As in the ATE-setting, we define $\hat{\nu}^1_{l,-k}$ and $\hat{\nu}^2_{l,-k}$ as the nuisance estimators based on data in $\mathcal{V}_{-k}$. The cross-fitted estimators are defined as 

\iffalse
\begin{align}
    \hat{\Gamma}_j^{l, CF} &= \frac{1}{n} \sum_{k=1}^K \sum_{i \in \mathcal{T}_k} \phi_{\Gamma_j^l}(O_i; \hat{\nu}_{l,-k}^2) \nonumber \\
    \hat{\chi}^{l, CF} &= \frac{1}{n} \sum_{k=1}^K \sum_{i \in \mathcal{T}_k} \phi_{\chi^l}(O_i; \hat{\nu}_{l,-k}^1) \nonumber \\
    \hat{\Omega}_j^{l, CF} &= \frac{\hat{\Gamma}_j^{l, CF}}{\hat{\chi}^{l, CF}} \nonumber
\end{align}
\fi

\begin{align}
    \hat{\Gamma}_j^{l, CF} = \frac{1}{n} \sum_{k=1}^K \sum_{i \in \mathcal{T}_k} \phi_{\Gamma_j^l}(O_i; \hat{\nu}_{l,-k}^2), \quad
    \hat{\chi}^{l, CF} = \frac{1}{n} \sum_{k=1}^K \sum_{i \in \mathcal{T}_k} \phi_{\chi^l}(O_i; \hat{\nu}_{l,-k}^1), \quad
    \hat{\Omega}_j^{l, CF} = \frac{\hat{\Gamma}_j^{l, CF}}{\hat{\chi}^{l, CF}}. \nonumber
\end{align}

Since the above estimators depend on the extended nuisance estimators, we have to make additional assumption in order to derive the desired asymptotic linearity. Accordingly, we have the following result: 
\begin{theorem}\label{thm:anbplp}
    Assume that for each fold $k=1,\ldots, K$ it holds that 
    \begin{itemize}
        \item[(i)] $\left( X_l - \hat{E}^l \right)^2 \leq M, \ a.s $ for all $n$ and some $M>0.$
        \item[(ii)] $\norm{\hat{\tau}_j^l - \tau_j^l} = o_p(n^{-1/4}).$
        \item[(iii)] $\norm{\hat{E}^l - E^l} = o_p(n^{-1/4})$.
    \end{itemize}
    Then, if assumption \ref{ass:nuis} holds for each $k$, it follows that $\hat{\Omega}_j^{l,CF}$ is asymptotically linear with influence function given by $\tilde{\psi}_{\Omega_j^l}$ and hence
    $$
    \sqrt{n}(\hat{\Omega}_j^{l,CF} - \Omega_j^{l}) \overset{d}{\rightarrow} \mathcal{N}(0, P\tilde{\psi}_{\Omega_j^l}^2).
    $$
\end{theorem}
\begin{proof}
    See Appendix \ref{app:asympLin}.
\end{proof}
Assumptions $(i)-(iii)$ in Theorem \ref{thm:anbplp} refer to the nuisance estimators related to the conditional distribution of $X_l$ given $X_{-l}$. Regarding assumption \ref{ass:nuisDouble} for ATE estimation, we discussed the double robustness properties of the cross-fitted estimator in relation to the convergence rates of the nuisance estimators, and we can add to that discussion the rates given in $(ii)$ and $(iii)$. We see that the estimator for our target parameter achieves parametric rates (asymptotic linearity) if the nuisance estimators related to $X_l|X_{-l}$ are estimated at $n^{-1/4}$-rate, adding to the notion of "double robustness". The rate in assumption $(iii)$ is known for many estimators, as it is an assumption on a typical regression estimator. Whether it is fulfilled depends on the type of the estimator used and possibly on the dimension $d$ in relation to $n$, but we note that the assumption is to be considered rather mild, allowing for many types of data-adaptive estimators (see e.g. the discussion in \cite{kennedydouble}, Section 4.3). For estimation of $\hat{\tau}_j^l$, we regress the CATE estimates $\left(\hat{\tau}_j(X_i)\right)_{i=1}^n$ onto $X_{-l} = \left(X_{i,-l}\right)_{i=1}^n$ in line with the approach suggested in \textcite{hinesVIM}, and \textcite{zm}. This approach constitutes a certain type of meta-learning and convergence rates related to $(ii)$ are generally less known compared to the regression in assumption $(iii)$. We refer to \textcite{hinesVIM} for a discussion of a specific meta-learner termed the DR-learner (\cite{kennedyOptimal}) for estimation of $\hat{\tau}_j^l$ (their analogy is termed $\hat{\tau}_s$) and convergence rates analogous to $(ii)$. \\ \\
As in the ATE setting, the variance, $P\tilde{\psi}^2_{\Omega_j^l}$, is estimated by the cross-fitted plugin estimator:
$$
\hat{\sigma}_{\Omega_j^l}^{2,CF} = \sum_{k=1}^K \frac{n_k}{n} \mathbb{P}_n^k\tilde{\psi}_{\Omega_j^l}(\hat{\nu}^2_{l,-k})^2,
$$
where (with some abuse of notation) we define 
$$
\tilde{\psi}_{\Omega_j^l}(\hat{\nu}^2_{l,-k})(O) = \frac{1}{\hat{\chi}^{l, CF}}\left(\phi_{\Gamma_j^l}(O; \hat{\nu}_{l,-k}^2) - \hat{\Gamma}_j^{l,CF} - \hat{\Omega}_j^{l,CF} 
\left(\phi_{\chi^l}(O; \hat{\nu}_{l,-k}^1) - \hat{\chi}^{l,CF} \right) \right).
$$
Because of scale sensitivity, the estimate of $\Omega_j^l$ may be of less interest than testing the null-hypothesis $H_0:\Omega_j^l = 0$. A test statistic for $H_0$ can be defined by  
$$
\mbox{TST}_1^l\equiv \frac{\hat{\Omega}_j^{l,CF}}{\sqrt{\hat{\sigma}_{\Omega_j^l}^{2,CF}/n}}
$$
that is asymptotically standard normal distributed. Lemma 1 in \textcite{zm} shows that the cross-fitted variance estimators considered in this paper, i.e., $\hat{\sigma}_{\Omega_j^l}^{2,CF}$ and $\hat{\sigma}_{\psi_j}^{2,CF}$ are consistent. The following corollary (analogous to Corollary 4 in \cite{zm}) gives the desired asymptotic properties of our test-statistic:

\begin{corollary}
Under the same setup as in Theorem \ref{thm:anbplp}, we have under the null-hypothesis, $H_0:\Omega_j^l=0$, that 
$$
%\frac{\hat{\Omega}_j^{l,CF}}{\sqrt{\hat{\sigma}_{\Omega_j^l}^{2,CF}/n}} 
\mbox{TST}_1^l\overset{D}{\longrightarrow} \mathcal{N}(0,1).
$$
\end{corollary}
\begin{proof}
    Since the variance estimator $\hat{\sigma}_{\Omega_j^l}^{2,CF}$ is consistent, Theorem \ref{thm:anbplp} together with Slutsky's theorem and an application of the delta method gives the result.    
\end{proof}

\section{Simulation study}\label{sec:sim}
We conduct simulation studies to investigate the proposed asymptotic properties of the estimators $\hat{\psi}_j^{CF}$ and $\hat{\Omega}_j^{l,CF}$ under two different nuisance estimator settings with and without cross-fitting (i.e. setting $K=1$). For all cross-fitted estimators, we set $K=10$. In one nuisance estimator setting we consider correctly specified (semi)parametric nuisance estimators and in the other we use completely nonparametric estimators via random forest. The parametric nuisance estimators adhere to assumption \ref{ass:nuis} and so we would expect the target parameter estimators to perform according to theory both with and without cross-fitting. In case of nonparametric estimators, random survival forest are shown to adhere to assumption \ref{ass:nuisCons} in \textcite{cuiCons}, but it is unclear to what extend they admit rates corresponding to \ref{ass:nuisDouble}. Furthermore, the nonparametric estimators do not in general belong to a Donsker class, and we therefore expect the cross-fitted estimators to perform more in line with the theory compared to the non-cross-fitted version (see \cite{cher} and \cite{kennedydouble} for a discussion on cross-fitted one-step estimators). 

We consider data generated from the following models:
\begin{itemize}
    \item $X_l \sim \text{Unif}[-1,1], \ l = 1, \ldots, 4$ 
    \item $\pi(1\mid X) = \text{expit}(0.5X_1 + 0.5X_2)$
    \item $\lambda_1(t\mid A, X) = 0.0025\cdot 2t^{2-1} \exp(-X_1 - X_2 - 0.2X_3 + A(0.5X_1 - 0.3X_2 - 2))$
    \item $\lambda_2(t\mid A, X) = 0.00025\cdot 2t^{2-1} \exp(-X_1 - X_2 - 0.2X_3 + A)$
    \item $\lambda_c(t\mid A, X) = 0.00025\cdot 2t^{2-1} \exp(-0.5X_1)$.
\end{itemize}
Note that the above hazard functions correspond to Cox models with baseline hazards given by a Weibull hazard, $bkt^{k-1}$, where $b$ and $k$ are the scale and shape parameter, respectively. We consider four sample size settings of $n=250, 500, 750, 1000$, and for each setting we run 1000 simulations. For each simulation we generate data according to the models above and estimate the target parameters according to specifications given below. 

\subsection{Average treatment effect}
For estimation of $\hat{\psi}_1^{CF}$ we choose the time-horsizon $t^*=30.$ The true ATE is approximately $\psi_1(P_0) = -9.6135$. We consider two nuisance setting for estimation of $\hat{\nu}$ (here dropping $k$ from the notation). In one setting $\hat{\nu}$ consists of correctly specified Cox models with corresponding Breslow estimators for the cumulative hazards, and a correctly specified logistic regression for the propensity score model. This setting will be abbreviated \textbf{cor} in tables and figures going forward. In the second setting we estimate the cumulative hazards by random survival forests (\cite{ishwaran}) as implemented in the \verb|R|-package \verb|randomForestSRC| with default tuning parameters (see \cite{rfsrc} for documentation), and we estimate the propensity score by random forest, again with the implementation and tuning parameters given by \verb|randomForestSRC|. This setting will be abreviated \textbf{RF} in tables and figures going forward. Furthermore, each setting will be given the suffix \textbf{CF} if cross-fitting is used.

\begin{table}[h!]
    \centering
    \begin{tabular}{|rlrrrr|}
  \hline
n & method & bias & SD & mean SE & coverage \\ 
  \hline
 250 & \textbf{cor} & -0.1521 & 0.9424 & 0.9746 & 0.9550 \\ 
     & \textbf{corCF} & -0.0494 & 0.9701 & 1.0340 & 0.9620 \\ 
     & \textbf{RF} & -0.3501 & 0.9020 & 0.5644 & 0.7520 \\ 
     & \textbf{RFCF} & 0.0953 & 1.2012 & 1.3429 & 0.9710 \\
     \hline
    500 & \textbf{cor} & -0.0980 & 0.6946 & 0.6893 & 0.9570 \\ 
     & \textbf{corCF} & -0.0538 & 0.7040 & 0.7084 & 0.9550 \\ 
     & \textbf{RF} & -0.3576 & 0.6750 & 0.3957 & 0.6800 \\ 
     & \textbf{RFCF} & -0.0383 & 0.7912 & 0.8799 & 0.9680 \\
     \hline
    750 & \textbf{cor} & -0.0665 & 0.5650 & 0.5633 & 0.9480 \\ 
     & \textbf{corCF} & -0.0377 & 0.5684 & 0.5735 & 0.9500 \\ 
     & \textbf{RF} & -0.3269 & 0.5538 & 0.3237 & 0.6750 \\ 
     & \textbf{RFCF} & -0.0678 & 0.6394 & 0.6970 & 0.9670 \\ 
     \hline
   1000 & \textbf{cor} & -0.0265 & 0.4724 & 0.4884 & 0.9540 \\ 
    & \textbf{corCF} & -0.0046 & 0.4754 & 0.4949 & 0.9570 \\ 
    & \textbf{RF} & -0.2800 & 0.4636 & 0.2808 & 0.6950 \\ 
    & \textbf{RFCF} & -0.0319 & 0.5383 & 0.5943 & 0.9710 \\  
   \hline
\end{tabular}
    \caption{Results of 1000 simulations of estimators of $\Psi_1$ with varying nuisance estimators, and with and without cross-fitting for sample sizes $n=250, 500, 750, 1000$. The abbreveations of the methods are read as follows: \textbf{cor} corresponds to the nuisance parameters $\Lambda_1, \Lambda_2, \Lambda_c, \pi$ estimated by correctly specified Cox and logistic regressions, and \textbf{RF} corresponds to the same parameters estimated by Random Forest. A suffix \textbf{CF} indicates that cross-fitting was employed in estimation of $\Psi_j$. The tables gives the bias, empirical standard deviation (SD), mean of the estimated standard error (mean SE), and coverage.}
    \label{tab:simate}
\end{table}

Table \ref{tab:simate} gives the results for ATE-estimation. For correctly specified nuisance estimators the bias decreases with $n$ and standard deviations decrease at an approximately $n^{1/2}$-rate. With a coverage around $0.95$, even for relatively small $n$, it looks as if the estimator follows the asymptotic distribution from Theorem \ref{thm:anATE}. Surprisingly though, cross-fitting seems to decrease the bias of the estimator with correctly specified nuisance parameters even further. Overall we find that the estimator performs as expected for correctly specified (semi)parametric nuisance estimators.

For the nuisance estimators using random forests, we see a non-vanishing bias for the non-cross-fitted estimators. Furthermore, the standard errors are underestimated compared to the empirical standard deviation of the estimators which, together with the bias result in undercoverage. When using cross-fitting together with random forests, the bias disappears on roughly the same order as the correctly specified parametric estimators (without cross-fitting), and the standard deviation of the estimator seem to be on roughly $\sqrt{n}$-rate, as with the correctly specified nuisance parameters. The standard errors seem to be slighty overestimated, though, resulting in a slight overcoverage. This might be due to the hyperparameter choices for the random forests.

\subsection{Best partially linear projection}
For estimation of $\hat{\Omega}_1^{l,CF}$, we set $t^*=30$, as for ATE-estimation. The true value of the target parameters are approximately $(\Omega_1^1, \Omega_1^2, \Omega_1^3, \Omega_1^4) = (4.949, 3.137, 0.737, 0)$. We also need estimates of $\hat{\tau}^l$ and $\hat{E}^l$, for which we will consider two settings. In the first, both $\hat{\tau}^l$ and $\hat{E}^l$ are estimated with a generalised additive model (GAM) including spline smoothing of each term but without interactions as implemented in the \verb|R|-package \verb|mgcv|, and in the second, each of the nuisance parameters are estimated with random forest (again with default tuning parameters from  \verb|randomForestSRC|). The GAM setting will be added to the correctly specified setting, \textbf{cor}, from earlier in tables and figures going forward, and the random forest setting will be added to the random forest setting, \textbf{RF}, from earlier.

\begin{figure}[h!]
    \centering
    \includegraphics[scale = 0.7]{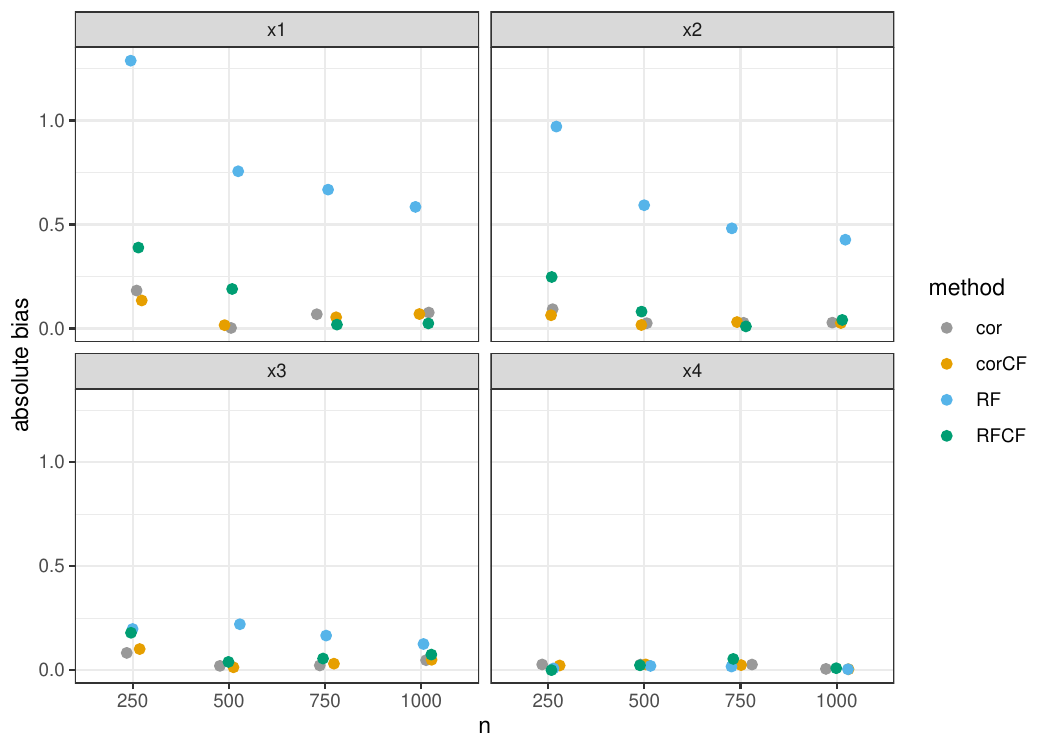}
    \caption{Results based on 1000 simulations of the estimators of $\Omega_1^l$ with for $l=1,\ldots,4$ with varying nuisance estimators and across sample sizes $n=250, 500, 750, 1000$. The plot shows the absolute bias of the estimators, where \textbf{cor} corresponds to the nuisance parameters $\Lambda_1, \Lambda_2, \Lambda_c, \pi$ estimated by correctly specified Cox and logistic regressions, and \textbf{RF} corresponds to the same parameters estimated by Random Forest. A suffix \textbf{CF} indicates that cross-fitting was employed. The true values are $(\Omega_1^1, \Omega_1^2, \Omega_1^3, \Omega_1^4) = (4.949, 3.137, 0.737, 0)$.}
    \label{fig:bias-omega}
\end{figure}

\begin{figure}[h!]
    \centering
    \includegraphics[scale = 0.7]{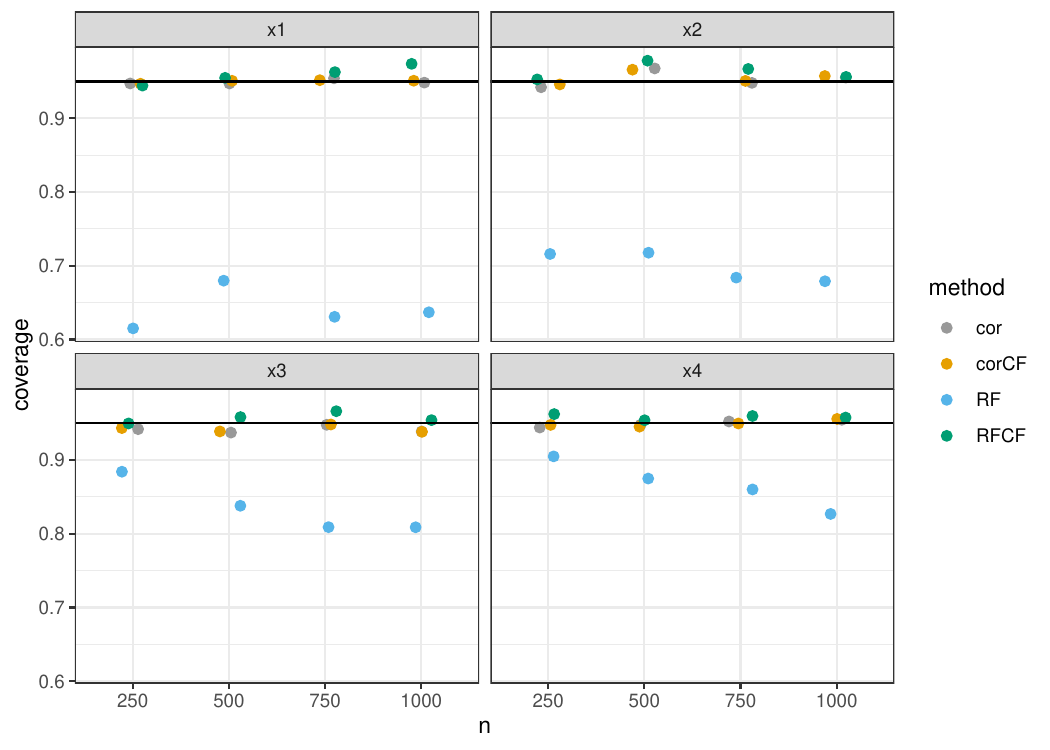}
    \caption{Results based on 1000 simulations of the estimators of $\Omega_1^l$ with for $l=1,\ldots,4$ with varying nuisance estimators and across sample sizes $n=250, 500, 750, 1000$. The plot shows coverage of the estimators, where \textbf{cor} corresponds to the nuisance parameters $\Lambda_1, \Lambda_2, \Lambda_c, \pi$ estimated by correctly specified Cox and logistic regressions, and \textbf{RF} corresponds to the same parameters estimated by Random Forest. A suffix \textbf{CF} indicates that cross-fitting was employed. The black line indicates a coverage of $0.95$.}
    \label{fig:coverage-omega}
\end{figure}

In Figure \ref{fig:bias-omega}, we see the absolute bias for estimation of $\hat{\Omega}_j^l$, $l = 1, \ldots, 4$, for the different nuisance settings. In general, we see that \textbf{cor} and \textbf{corCF} perform similarly across all sample sizes and across all $l$, with a bias converging to zero. The \textbf{RFCF}-setting performs slightly worse than \textbf{cor} and \textbf{corCF} for small $n$, but has approximately similar performance for large $n$, whereas \textbf{RF} has a large bias for large enough values of $\Omega_1^l$. Generally, the estimators seem to perform as we would expect according to Theorem \ref{thm:anbplp} in terms of bias. 
The coverage of the estimators are presented in Figure \ref{fig:coverage-omega}. The settings \textbf{cor}, \textbf{corCF} and \textbf{RFCF} all exhibit approximately nominal coverage across $n$ and $l$, whereas \textbf{RF} has poor coverage. Again, this is in line with our expectations. Figure \ref{fig:type1} gives the estimated probability of rejecting $H: \Omega_1^4 = 0$, i.e. the type-1 error (since $\Omega_1^4 = 0$ in our data generating mechanism), together with Monte Carlo confidence intervals. The type-1 error is approximately $0.05$ for all $n$, except for \textbf{RF} where the type 1-error increases with $n$. 

\begin{figure}[h!]
    \centering
    \includegraphics[scale = 0.7]{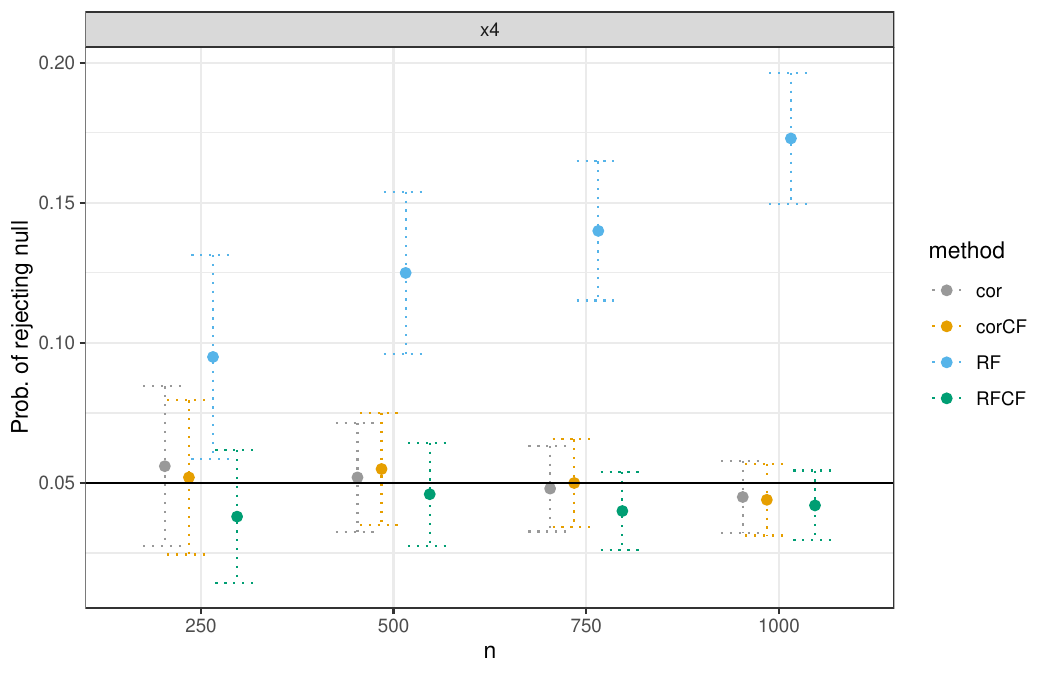}
    \caption{Results based on 1000 simulations of the test statistic corresponding to the test $H_0:\Omega_1^4 = 0$ with varying nuisance estimators and across sample sizes $n=250, 500, 750, 1000$. The plot shows probability of rejecting $H_0$, which equals the type-1 error as $\Omega_1^4=0$. \textbf{cor} corresponds to the nuisance parameters $\Lambda_1, \Lambda_2, \Lambda_c, \pi$ estimated by correctly specified Cox and logistic regressions, and \textbf{RF} corresponds to the same parameters estimated by Random Forest. A suffix \textbf{CF} indicates that cross-fitting was employed.}
    \label{fig:type1}
\end{figure}

\begin{figure}[h!]
    \centering
    \includegraphics[scale = 0.7]{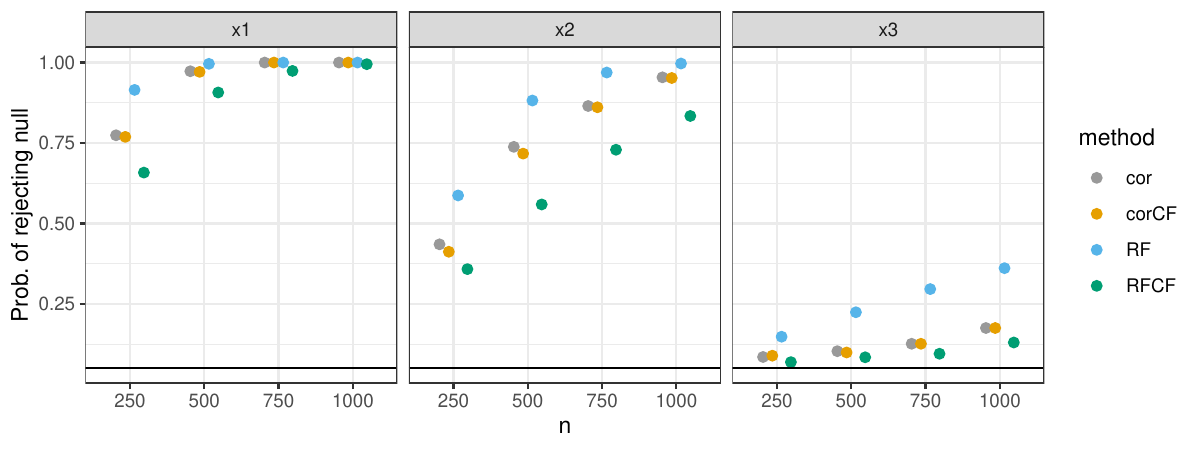}
    \caption{Results based on 1000 simulations of the test statistic corresponding to the test $H_0:\Omega_1^l = 0$, for $l=1,2,3$, with varying nuisance estimators and across sample sizes $n=250, 500, 750, 1000$. The plot shows probability of rejecting $H_0$, which corresponds to the power of the test as $\Omega_1^l > 0$ for $l=1,2,3$. \textbf{cor} corresponds to the nuisance parameters $\Lambda_1, \Lambda_2, \Lambda_c, \pi$ estimated by correctly specified Cox and logistic regressions, and \textbf{RF} corresponds to the same parameters estimated by Random Forest. A suffix \textbf{CF} indicates that cross-fitting was employed.}
    \label{fig:power}
\end{figure}
Lastly, Figure \ref{fig:power} shows the probability of rejecting $H: \Omega_1^l = 0$, $l = 1,2,3$, which correspond the power of the test. Interestingly, using data-adaptive estimation of the nuisance parameters seem to decrease the power of the test $TST_1^l$. 

\iffalse
But since the standard errors used to construct the tests are estimated with a plug-in estimator, it is perhaps not surprising that we do not get similar performances between, e.g., \textbf{cor} and \textbf{RFCF}, since the test are not debiased for machine-learning nuisance estimators.
\fi

\section{Application}\label{sec:application}
To demonstrate the methods outlined in the previous sections, we consider the study by \textcite{kessing}. In this study, the potentially non-response to 17 different antidepressants are compared based on data from the Danish national registers. Patients enter the study at their first diagnosis with major depressive disorder from a psychiatric hospital. Their treatment, in terms of a specific antidepressant, is defined as the first purchase of an antidepressant after discharge from the hospital, which also determines the index date. The main outcome is time to non-response, defined as a switch to or add-on of another antidepressant or antipsychotic medicine or readmission to a psychiatric hospital with a major depressive disorder. Competing risk for the time to non-response was admission to a psychiatric hospital with a higher order psychiatric diagnosis (bipolar disorder, schizophrenia or organic mental disorder) or death. The 17 antidepressants are categorised into six groups (SSRI, NARI, SNRI, NaSSA, TCA, and others) and within each group a reference drug is chosen to which the other drugs in that group are compared. The estimand for each comparison is the average treatment effect on the risk of non-response within two years after index date, i.e., it is defined as $\E\{F_1(t\mid A=1, X) - F_1(t\mid A=0, X)\}$, where $F_1$ is the conditional cumulative incidence function for non-response and $A=0$ denotes the reference drug and $A=1$ is the comparitor. The study includes patients from 1995-2018 and not all of the antidepressants considered were available on the marked in the entire period. Accordingly, for each comparison, a minimum date is set for which both drugs in comparison were available and all patients with index dates prior to the minimum date are excluded. 

\begin{table}[h!]
    \centering
    \begin{tabular}{|m{0.2\textwidth}|m{0.6\textwidth}|}
        \hline
        \textbf{Annotation} & \textbf{Explanation} \\
        \hline
        age & age in years at index-date \\%age-group: $(<30, 30–50, 50–70, >70)$ at index-date \\
        \hline 
        sex & female, male \\
        \hline
        \multicolumn{2}{|m{0.8\textwidth}|}{\textbf{Secondary diagnosis from the psychiatric hospital at inclusion. The annotations reflect the ICD-10 codes used in the definition:}} \\
        \hline
        F10-19 & other psychiatric disorders\\
        \hline 
        F40-48 & neurotic, stress-related and somatoform disorders  \\
        \hline 
        F60-69 & personalitiy disorders \\
        \hline
        \multicolumn{2}{|m{0.8\textwidth}|}{\textbf{Diagnosis with somatic disorder within 10 years prior to index date. The annotation are given in form of the corresponding ICD-10 chapter:}} \\
        \hline
        I & infections \\
        \hline
        II &  neoplasms  \\
        \hline
        III & blood diseases  \\
        \hline
        IV+IX+X & endocrine, nutritional, and metabolic diseases
and diseases of the circulatory or respiratory system \\
        \hline
        VI-VIII & diseases of the nervous system, eye and ear \\
        \hline
        XI & diseases of the digestive system \\
        \hline
        XII & diseases of the skin and subcutaneous tissue \\
        \hline
        XIII & diseases if the musculoskeletal system \\
        \hline
        XIX & physical lesions and poisoning) \\
        \hline
    \end{tabular}
    \caption{Confounders in \textcite{kessing}}
    \label{tab:confounders}
\end{table}

For the sake of illustration we constrict ourselves to the comparison of Setraline (reference drug, $n=14416$) and Escitalopram (comparitor, $n=7508$). \textcite{kessing} used the G-formula with $F_1$ estimated by cause-specific Cox regressions (\cite{ozenne}) to estimate the average treatment effect. To control for confounding, the Cox-regressions were adjusted for the covariates in Table \ref{tab:confounders} and the ATE was estimated to be 0.10 (0.09, 0.12), that is, the probability of non-response was 0.1 higher amongst patients treated with Escitalopram within two years after treatment initiation. For comparison, instead of defining the treatment effect through the cumulative incidence function, $F_1$, we consider estimation of the ATE and the best partially linear projection based on the number of life-years-lost estimands defined in Section \ref{sec:estimands}. That is, we consider estimation of $\psi_1$ and $\Omega_1^l$ with $t^* = 730.5 \ days$ (2 years). We include the same confounders as in \textcite{kessing} with the exception that age is included as a numeric variable instead of a categorised version. The target parameters are estimated with the cross-fitted estimators described in Section \ref{sec:est} with all nuisance parameters estimated by random forests (as described in the simulation study) and $K=10$ folds.

The ATE is estimated to 48.96 (40.02, 57.90). The interpretation here is that patients on Setraline on average lost 49 "healthy" days less before two years after treatment initiation due to non-response compared to patients who start on Escitalopram, where "healthy" is meant as time without a non-response event or a competing event. Table \ref{tab:vim-application} shows the estimates $\hat{\Omega}_1^{l,CF}$, for $l$ given by each confounder in Table \ref{tab:confounders}, together with a p-value associated with the test statistic $\mbox{TST}_1^l$.
%$$\frac{\hat{\Omega}_1^{l,CF}}{\sqrt{\hat{\sigma}_{\Omega_1^l}^{CF}/n}}.$$ 

\nothere{
\begin{table}[h!]
    \centering
    \begin{tabular}{|c|c|c|}
    \hline
        & $\hat{\Omega}_1^{l,CF}$ & p-value \\
        \hline
    age & -0.85 & 0.029 \\
    \hline 
    sex & 18.69 & 0.035 \\
    \hline
    XIII & -28.26 & 0.169 \\
    \hline 
    XI & -16.64 & 0.274 \\
    \hline
    F60-69 & 15.05 & 0.297 \\
    \hline
    II & 17.35 & 0.364 \\
    \hline
    VI-VIII & 8.88 & 0.403 \\
    \hline
    IV+IX+X & 4.38 & 0.733 \\
    \hline
    F40-48 & 5.21 & 0.758 \\
    \hline
    III & -4.80 & 0.818 \\
    \hline
    XIX & -2.85 & 0.826 \\
    \hline
    XII & -2.00 & 0.875 \\
    \hline
    I & -0.41 & 0.981 \\
    \hline
    F10-19 & 0.23 & 0.982 \\
    \hline
    \end{tabular}
    \caption{Estimates of $\hat{\Omega}_1^{l,CF}$ ranked according to the p-value associated with the test of $H:\Omega_1^l = 0$, for $l$ ranging over different covariates. Random Forest was used for all nuisance parameter estimators. The data comes from the study \textcite{kessing}, and the outcome is time to non-response, which is defined as a switch in psychiatric treatment or re-hospitalization at psychiatric ward. The treatment effect was defined as the difference in the number of healthy days lost (days without switch of treatment or re-hospitalization) due to non-response before two years after treatment initiation between Escitalopram and Setraline.}
    \label{tab:vim-application}
\end{table}
}

\begingroup

\setlength{\tabcolsep}{1.3pt} % Default value: 6pt
\renewcommand{\arraystretch}{0.3} % Default value: 1

\begin{table}[h!]
\centering
\begin{tabular}{c|ccccccccccccccc}
%\hline
&age & sex&{\small XIII} & {\small XI} &{\small F60-69} &{\small II} & {\small VI-VIII} &{\small IV+IX+X} &{\small F40-48} & {\small III} & {\small XIX} & {\small XII} &{\small I} & {\small F10-19}
 \\\hline
 & & & & & & &  & & &  &  &  &  &  \\
$\hat{\Omega}_1^{l,CF}$ &-0.85&18.7&-28.3&-16.7&15.1&17.4&8.88&4.38&5.21&-4.80&-2.85&-2.00&-0.41&0.23\\
& & & & & & &  & & &  &  &  &  &  \\
 p-value& 0.03&0.04 &0.17&0.27&0.30&0.36&0.40&0.73&0.76&0.82&0.83&0.88&0.98&0.98\\
 & & & & & & &  & & &  &  &  &  &  \\
\hline
\end{tabular}
    \caption{Estimates of $\hat{\Omega}_1^{l,CF}$ ranked according to the p-value associated with the test of $H:\Omega_1^l = 0$, for $l$ ranging over different covariates. Random Forest was used for all nuisance parameter estimators. The data comes from the study \textcite{kessing}, and the outcome is time to non-response, which is defined as a switch in psychiatric treatment or re-hospitalization at psychiatric ward. The treatment effect was defined as the difference in the number of healthy days lost (days without switch of treatment or re-hospitalization) due to non-response before two years after treatment initiation between Escitalopram and Setraline.}
    \label{tab:vim-application}
\end{table}   

\endgroup

The p-values indicate that potential treatment effect heterogeneity can be attributed to sex and age, while the treatment effect does not seem to vary across any of the other variables. Specifically, since $\Omega_1^l$ is defined as the projection of the CATE function onto the partially linear model, the estimates related to sex and age can be interpreted as regression coefficients. The CATE function is defined as the difference in number of healthy days lost due to non-response between Escitalopram and Setraline, and the estimate $\hat{\Omega}_1^{sex, CF}=18.69\ $ then corresponds to the treatment effect being larger among women compared to men, i.e., the difference in number of healthy days lost due to non-response between Escitalopram and Setraline was larger among women. This interpretation is of course relying on the partially linear model to hold for the CATE function, but as the parameter still measures the association between $\tau$ and sex, when the partially linear model does not hold, we would still conclude, that the treatment effect is larger among women.

\iffalse
It should be noted, however, that the p-values in table \ref{tab:vim-application} are unadjusted and does not account for the multiple testing. The p-values are kept this way for illustrative purposes, and, clearly, any conservative adjustements, like Bonferroni, will result in the conclusion that their is no treatment effect heterogeneity. On the other hand, the correlation of the tests should be taken into account when making adjustments of the p-values, as to control the global type-1 error in a non-overconservative way. Furthermore, as the variable importance measures are more exploratory rather than confirmatory in nature, one could also make an argument for controlling the false discovery rate rather than the type-1 error. 
\fi

\section{Discussion}\label{sec:disc}
In this paper, we have introduced the causal effect of a treatment on the number of life-years lost due to a specific event. We have shown that the ATE and the CATE are identifiable from the observed data under common assumptions found throughout the causal inference and survival literature. Different measures of treatment effect in the presence of competing risk are available in causal inference (\cite{heleneMark}, \cite{heleneFrank}, \cite{ozenne}, \cite{torbenMats}) and the treatment effect studied in this paper adds a new interpretability compared to existing variants. One advantage is that the treatment effect is defined on the time scale of the study and it thus provides a quantity that is easy to communicate to non-statisticians, whereas treatment effects based on the cumulative incidence function are harder to communicate. Furthermore, as the treatment effect can be written as a difference of integrated cumulative incidence functions, it is not as sensitive to the choice of time horizon in terms of detecting an effect of treatment. As is common when assessing the treatment effect in the presence of competing risk, the effect of a treatment on the number of life years lost due to a specific event depends on the effect of the treatment on both the hazard of the event of interest and on the competing event. As such, one can in principle conclude that there is an effect of treatment, even when all of the effect is driven by the effect on the competing event. \textcite{torbenMats} provide a measure of separable treatment effects based on the cumulative incidence function, which allow one to estimate the effect of treatment only driven by the intensity of the event of interest under additional causal assumptions. An interesting avenue for future research is to extend their method to the number of life years lost due to a specific event. \\ \\
We have provided an estimator of the ATE based on semiparametric efficiency theory, which allows for data-adaptive estimation of the nuisance parameters. The estimator is efficient in the nonnparametric model with variance given by the efficient influence function. One of the assumptions needed to ensure the asymptotic results relies on convergence of a remainder term on $n^{-1/2}$-rate (assumption \ref{ass:nuisDouble}), which is a reminiscent of a similar assumption in the causal inference literature with censored data (\cite{westling}, \cite{heleneFrank}). Without assuming absolute continuity of the hazard estimators, it is difficult to  obtain an equivalent double rate-robustness property as is seen in the literature on uncensored data (e.g. \cite{kennedydouble}, \cite{hinesDemystifying}, \cite{cher}, \cite{vanRose}). Accordingly, we conducted a simulation study, where the nuisance parameters were estimated by different variants of random forests, which confirmed that the estimator performed according to asymptotic results when using data-adaptive nuisance estimators. \\ \\
Lastly, we extended a measure of treatment effect heterogeneity, termed the best partially linear projection (\cite{zm}), to the CATE-function defined on the number of life-years lost due to a specific event. The measure asserts the importance of a given covariate on the treatment effect, but with competing risk the interpretation is more delicate compared to the survival setting. When the effect of treatment on the competing event is large, one can imagine scenarios where the importance of one covariate on the CATE is driven by the effect on the competing event, and a ranking of importance (shown in Section \ref{sec:application}) based on the event of interest might be misleading. In such scenarios, one can switch the event of interest and competing event and make a separate ranking of the covariates to get a full picture. 

\clearpage
%\nocite{*}
\printbibliography

@article{overgaard2019,
  title={Pseudo-observations under covariate-dependent censoring},
  author={Overgaard, Morten and Parner, Erik Thorlund and Pedersen, Jan},
  journal={Journal of Statistical Planning and Inference},
  volume={202},
  pages={112--122},
  year={2019},
  publisher={Elsevier}
}

@article{andersenParner,
  title={Causal inference in survival analysis using pseudo-observations},
  author={Andersen, Per K and Syriopoulou, Elisavet and Parner, Erik T},
  journal={Statistics in medicine},
  volume={36},
  number={17},
  pages={2669--2681},
  year={2017},
  publisher={Wiley Online Library}
}

@article{andersen,
  title={Decomposition of number of life years lost according to causes of death},
  author={Andersen, Per Kragh},
  journal={Statistics in medicine},
  volume={32},
  number={30},
  pages={5278--5285},
  year={2013},
  publisher={Wiley Online Library}
}

@article{gill,
  title={A survey of product-integration with a view toward application in survival analysis},
  author={Gill, Richard D and Johansen, Soren},
  journal={The annals of statistics},
  volume={18},
  number={4},
  pages={1501--1555},
  year={1990},
  publisher={Institute of Mathematical Statistics}
}

@article{zm,
  title={Variable importance measures for heterogeneous treatment effects with survival outcome},
  author={Ziersen, Simon C and Martinussen, Torben},
  journal={arXiv preprint arXiv:2412.11790},
  year={2024}
}

@article{robins,
  title={A new approach to causal inference in mortality studies with a sustained exposure period—application to control of the healthy worker survivor effect},
  author={Robins, James},
  journal={Mathematical modelling},
  volume={7},
  number={9-12},
  pages={1393--1512},
  year={1986},
  publisher={Elsevier}
}

@book{ABGK,
  title={Statistical models based on counting processes},
  author={Andersen, Per K and Borgan, Ornulf and Gill, Richard D and Keiding, Niels},
  year={1993},
  publisher={Springer Science \& Business Media}
}

@book{ms,
  title={Dynamic regression models for survival data},
  author={Martinussen, Torben and Scheike, Thomas H},
  volume={1},
  year={2006},
  publisher={Springer}
}

@book{Vaart,
  title={Asymptotic statistics},
  author={van der Vaart, Aad W},
  volume={3},
  year={2000},
  publisher={Cambridge university press}
}

@article{kennedydouble,
  title={Semiparametric doubly robust targeted double machine learning: a review},
  author={Kennedy, Edward H},
  journal={arXiv preprint arXiv:2203.06469},
  year={2022}
}

@article{cuiCons,
  title={Consistency of survival tree and forest models: splitting bias and correction},
  author={Cui, Yifan and Zhu, Ruoqing and Zhou, Mai and Kosorok, Michael},
  journal={Statistica Sinica},
  volume={32},
  number={3},
  pages={1245--1267},
  year={2022},
  publisher={JSTOR}
}

@article{cher,
  title={Double/debiased machine learning for treatment and structural parameters},
  author={Chernozhukov, Victor and Chetverikov, Denis and Demirer, Mert and Duflo, Esther and Hansen, Christian and Newey, Whitney and Robins, James},
  year={2018},
  publisher={Oxford University Press Oxford, UK}
}

@article{rfsrc,
  title={Package ‘randomForestSRC’},
  author={Ishwaran, Hemant and Kogalur, Udaya B and Kogalur, Maintainer Udaya B},
  journal={breast},
  volume={6},
  number={1},
  pages={854},
  year={2023}
}

@article{ishwaran,
  title={Random survival forests},
  author={Ishwaran, Hemant and Kogalur, Udaya B and Blackstone, Eugene H and Lauer, Michael S},
  year={2008}
}

@book{vanRose,
  title={Targeted learning},
  author={van der Laan, Mark J and Rose, Sherri},
  year={2011},
  publisher={Springer}
}

@book{bickel,
  title={Efficient and adaptive estimation for semiparametric models},
  author={Bickel, Peter J and Klaassen, Chris AJ and Bickel, Peter J and Ritov, Ya’acov and Klaassen, J and Wellner, Jon A and Ritov, YA'Acov},
  volume={4},
  year={1993},
  publisher={Springer}
}

@article{heleneFrank,
  title={Estimation of time-specific intervention effects on continuously distributed time-to-event outcomes by targeted maximum likelihood estimation},
  author={Rytgaard, Helene CW and Eriksson, Frank and van der Laan, Mark J},
  journal={Biometrics},
  year={2023},
  publisher={Wiley Online Library}
}

@article{heleneMark,
  title={Targeted maximum likelihood estimation for causal inference in survival and competing risks analysis},
  author={Rytgaard, Helene CW and van der Laan, Mark J},
  journal={Lifetime Data Analysis},
  volume={30},
  number={1},
  pages={4--33},
  year={2024},
  publisher={Springer}
}

@article{westling,
  title={Inference for treatment-specific survival curves using machine learning},
  author={Westling, Ted and Luedtke, Alex and Gilbert, Peter B and Carone, Marco},
  journal={Journal of the American Statistical Association},
  number={just-accepted},
  pages={1--26},
  year={2023},
  publisher={Taylor \& Francis}
}

@article{heleneThomas,
  title={Continuous-time targeted minimum loss-based estimation of intervention-specific mean outcomes},
  author={Rytgaard, Helene CW and Gerds, Thomas A and van der Laan, Mark J},
  journal={The Annals of Statistics},
  volume={50},
  number={5},
  pages={2469--2491},
  year={2022},
  publisher={Institute of Mathematical Statistics}
}

@article{ozenne,
  title={On the estimation of average treatment effects with right-censored time to event outcome and competing risks},
  author={Ozenne, Brice Maxime Hugues and Scheike, Thomas Harder and St{\ae}rk, Laila and Gerds, Thomas Alexander},
  journal={Biometrical Journal},
  volume={62},
  number={3},
  pages={751--763},
  year={2020},
  publisher={Wiley Online Library}
}

@article{kessing,
  title={Comparative responses to 17 different antidepressants in major depressive disorder: Results from a 2-year long-term nation-wide population-based study emulating a randomized trial},
  author={Kessing, Lars Vedel and Ziersen, Simon Christoffer and Andersen, Frederik Mølkjær and Gerds, Thomas and Budtz-Jørgensen, Esben},
  journal={Acta Psychiatrica Scandinavica},
  year={2024},
  publisher={Wiley Online Library}
}

@article{hinesDemystifying,
  title={Demystifying statistical learning based on efficient influence functions},
  author={Hines, Oliver and Dukes, Oliver and Diaz-Ordaz, Karla and Vansteelandt, Stijn},
  journal={The American Statistician},
  volume={76},
  number={3},
  pages={292--304},
  year={2022},
  publisher={Taylor \& Francis}
}

@article{munch,
  title={Estimating conditional hazard functions and densities with the highly-adaptive lasso},
  author={Munch, Anders and Gerds, Thomas A and van der Laan, Mark J and Rytgaard, Helene CW},
  journal={arXiv preprint arXiv:2404.11083},
  year={2024}
}

@article{hinesVIM,
  title={Variable importance measures for heterogeneous causal effects},
  author={Hines, Oliver and Diaz-Ordaz, Karla and Vansteelandt, Stijn},
  journal={arXiv preprint arXiv:2204.06030},
  year={2022}
}

@article{kennedyOptimal,
  title={Towards optimal doubly robust estimation of heterogeneous causal effects},
  author={Kennedy, Edward H},
  journal={Electronic Journal of Statistics},
  volume={17},
  number={2},
  pages={3008--3049},
  year={2023},
  publisher={The Institute of Mathematical Statistics and the Bernoulli Society}
}

@article{torbenMats,
  title={Estimation of separable direct and indirect effects in continuous time},
  author={Martinussen, Torben and Stensrud, Mats Julius},
  journal={Biometrics},
  volume={79},
  number={1},
  pages={127--139},
  year={2023},
  publisher={Wiley Online Library}
}
\newpage
\appendix
\section{Identification of causal estimands}\label{app:identification}
We present an argument for the identification results in the main text. The argument follows usual steps in the causal inference literature on censored data, and it is based of a combination of the G-formula (\cite{robins}) and identification results from the literature on survival analysis (see e.g. \cite{ABGK} and \cite{ms}).

As noted in \textcite{andersen}, $T_j$ is improper due to $P(T_j = \infty) > 0$, but the random variable $T_j\wedge t^*$ is proper with an expectation given by
$$
E\{T_j\wedge t^*\} = t^* - \int_0^{t^*} F_j(s) \dd s.
$$
Hence 
$$
E(T_j\wedge t^* \mid a, x ) = t^* - \int_0^{t^*} F_j(s\mid a, x) \dd s.
$$
which is identified in the observed data under assumption \ref{ass:indepCens} and \ref{ass:pos} (\cite{ABGK}). For the ATE, this allows us to write
\begin{align*}
    & \phantom{= \ } E_0\{Y_j^1(t^*) - Y_j^0(t^*)\} \\
  &= E_0\{\E(Y_j^1(t^*) - Y_j^0(t^*)\mid X)\} \\
  \overset{ass. \ref{ass:exch}}&{=}  E_0\{\E(Y_j^1(t^*)\mid A = 1, X) - \E(Y_j^1(t^*)\mid A = 0, X)\} \\
  \overset{ass. \ref{ass:cons}}&{=}  E_0\{\E(Y_j(t^*)\mid A = 1, X) - \E(Y_j(t^*)\mid A = 0, X)\} \\
  \overset{ass. \ref{ass:indepCens}}&{=}  E_0\{ L_j(0,t^*\mid A=1, X) - L_j(0,t^*\mid A=0, X)\} 
\end{align*}
where assumption \ref{ass:pos} ensures that all conditional distributions are well defined. We note that CATE is identified by the same arguments.
\section{Derivation of influence functions}\label{app:eif}
In the following we consider the parametric submodel $P_\epsilon = \epsilon Q + (1-\epsilon)P$, where $Q$ is the Dirac measure with pointmass in the observation $O = (\tilde{T}, \tilde{\Delta}, A, X)$, and we define the operator $\partial_\epsilon$ with $\partial_\epsilon f_\epsilon = \frac{d}{d \epsilon}|_{\epsilon = 0} f_\epsilon$.
\begin{alemma}\label{gat:dlambda}
    The Gateaux derivative of $\Lambda_j(ds\mid a, x)$ is given by
    \begin{align}
    \partial_\epsilon \Lambda_{j,\epsilon}(ds\mid a, x) = \frac{1}{P(\tilde{T} \geq s, a, x)}\left(Q(ds, \Delta = j, a, x) - \mathbb{1}(\tilde{T}\geq s, a, x)\Lambda_j(ds\mid a, x)\right)
    \end{align}
    where 
    $$
    P(\tilde{T} \geq s, a, x) = \sum_{\delta = 0}^2\int_s^\infty P(ds, \delta, a, x).
    $$
\end{alemma}
\begin{proof}
Observe
    \begin{align*}
       & \partial_\epsilon \Lambda_{j,\epsilon}(ds\mid a, x) \\
     = & \partial_\epsilon \frac{P_\epsilon(ds, \Delta = j, a, x)}{P_\epsilon(\tilde{T}\geq s, a, x)} \\
     = & \frac{Q(ds, \Delta = j, a, x) - P(ds, \Delta = j, a, x)}{P(\tilde{T}\geq s, a, x)} - \partial_\epsilon P_\epsilon(\tilde{T}\geq s, a, x) \frac{P(ds, \Delta = j, a, x)}{P(\tilde{T}\geq s, a, x)^2} \\
     = & \frac{Q(ds, \Delta = j, a, x) - P(ds, \Delta = j, a, x)}{P(\tilde{T}\geq s, a, x)} \\
       & - \sum_{\delta = 0}^2\left(\mathbb{1}(\tilde{T}\geq s, \delta, a, x) - P(\tilde{T}\geq s, \delta, a, x) \right)\frac{P(ds, \Delta = j, a, x)}{P(\tilde{T}\geq s, a, x)^2} \\
     = & \frac{1}{P(\tilde{T}\geq s, a, x)}\left( Q(ds, \Delta = j, a, x) - \mathbb{1}(\tilde{T}\geq s, a, x)\Lambda_j(ds\mid a, x) \right)
    \end{align*}
\end{proof}

\begin{alemma}\label{gat:S}
    The Gateaux derivative of $S(s\mid a,x)$ is given by 
    \begin{align}
        \partial_\epsilon S_\epsilon(s\mid a, x) = -S(s\mid a, x)\frac{\mathbb{1}(A=a, X=x)}{\pi(a\mid x), f(x)} \int_0^s \frac{\dd M_1(u\mid a, x) + \dd M_2(u\mid a,x)}{P(\tilde{T}\geq s, a, x)}
    \end{align}
\end{alemma}
\begin{proof}
    First we note that
    \begin{align}
        \partial_\epsilon \Lambda_{j, \epsilon}(s\mid a, x) =&  \int_0^s \partial_\epsilon \Lambda_{j,\epsilon}(ds\mid a, x) \nonumber \\
        =& \frac{\mathbb{1}(A=a, X=x)}{\pi(a\mid x)f(x)} \int_0^s \frac{\dd M_j(u\mid a, x)}{P(\tilde{T}\geq s\mid a, x)} \label{gat:Lambda}
    \end{align}
    by lemma \ref{gat:dlambda}. Then
    \begin{align*}
        \partial_\epsilon S(s\mid a, x) &= \partial_\epsilon \exp(-\left[ \Lambda_{1, \epsilon}(s\mid a,x) + \Lambda_{2, \epsilon}(s\mid a,x) \right]) \\
        &= -S(s\mid a, x) \partial_\epsilon \left[ \Lambda_{1, \epsilon}(s\mid a,x) + \Lambda_{2, \epsilon}(s\mid a,x) \right].
    \end{align*}
    Applying \eqref{gat:Lambda} gives the result.
\end{proof}

\begin{alemma}\label{lem:gatLj}
    The Gateaux derivative of $L_j(0,t^*\mid a, x)$ is given by
    \begin{align}
        \partial_\epsilon L_{j, \epsilon}(0,t^*\mid a, x) = \frac{\mathbb{1}(A=a, X=x)}{\pi(a\mid x)f(x)} \sum_{i=1,2} \int_0^{t^*} \frac{H_{ij}(s,t^*\mid a, x)}{S_C(s\mid a, x)} \dd M_i(s\mid a, x)
    \end{align}
    where
    $$
    H_{ij}(s,t\mid a,x) = \int_s^t \mathbb{1}(i=j) + \frac{F_j(s\mid a, x) - F_j(u\mid a, x)}{S(s\mid a, x)} \dd u.
    $$
\end{alemma}
\begin{proof}
    First we note that
    \begin{align} \label{gat:Fjtmp}
        \partial_\epsilon F_{j, \epsilon}(t\mid a, x) = \int_0^t \partial_\epsilon S_\epsilon(s\mid a, x)\Lambda_j(ds\mid a, x) + \int_0^t \partial_\epsilon S(s\mid a, x)\Lambda_{j, \epsilon}(ds\mid a, x).
    \end{align}
    For the first term in \eqref{gat:Fjtmp}, observe
    \begin{align}
    &\int_0^t \partial_\epsilon S_\epsilon(s\mid a, x)\Lambda_j(ds\mid a, x) \\
    = & \frac{ \mathbb{1}(A=a, X=x)}{\pi(a\mid x)f(x)}\left[\int_0^t \int_0^s \frac{-S(s\mid a, x)}{P(\tilde{T}\geq u \mid a, x)} \dd M_1(u\mid a, x)\Lambda_j(s\mid a, x) \right. \nonumber \\
    &\left. \phantom{\frac{ \mathbb{1}(A=a, X=x)}{\pi(a\mid x)f(x)}} + \int_0^t \int_0^s \frac{-S(s\mid a, x)}{P(\tilde{T}\geq u \mid a, x)}\dd M_2(u\mid a, x)\Lambda_j(s\mid a, x)   \right] \nonumber
    \end{align}
    by lemma \ref{gat:S}, and for the second term
    \begin{align}
        &\int_0^t \partial_\epsilon S(s\mid a, x)\Lambda_{j, \epsilon}(ds\mid a, x) \\
        =& \int_0^t \frac{S(s\mid a,x)}{P(\tilde{T} \geq s, a, x)}\left(Q(ds, \Delta = j, a, x) - \mathbb{1}(\tilde{T}\geq s, a, x)\Lambda_j(ds\mid a, x)\right) \nonumber \\
        =& \frac{\mathbb{1}(A=a, X=x)}{\pi(a\mid x)f(x)}\int_0^t \frac{\dd N_j(s)}{S_C(s\mid a,x)} - \frac{\mathbb{1}(\tilde{T}\geq s) \Lambda_j(ds\mid a, x)}{S_C(s\mid a,x)} \nonumber \\
        =& \frac{\mathbb{1}(A=a, X=x)}{\pi(a\mid x)f(x)} \int_0^t \frac{\dd M_j(s\mid a, x)}{S_C(s\mid a,x)} \nonumber
    \end{align}
    by lemma \ref{gat:dlambda}. Plugging into \eqref{gat:Fjtmp} gives
    \begin{align} \label{gat:Fj}
        & \partial_\epsilon F_{j, \epsilon}(t\mid a, x) \nonumber \\
      = &  \frac{\mathbb{1}(A=a, X=x)}{\pi(a\mid x)f(x)} \left[ \int_0^t \frac{\int_u^t-S(s\mid a, x) \Lambda_j(ds\mid a, x)}{P(\tilde{T}\geq u\mid a, x)} \dd M_1(u\mid a, x) \right. \nonumber \\
        & \phantom{\frac{\mathbb{1}(A=a, X=x)}{\pi(a\mid x)f(x)} } + \int_0^t \frac{\int_u^t-S(s\mid a, x) \Lambda_j(ds\mid a, x)}{P(\tilde{T}\geq u\mid a, x)} \dd M_2(u\mid a, x) + \left. \int_0^t \frac{\dd M_j(s\mid a, x)}{S_C(s\mid a,x)} \right] \nonumber \\
      = & \frac{\mathbb{1}(A=a, X=x)}{\pi(a\mid x)f(x)} \left[ \int_0^t \frac{F_j(u\mid a,x) - F_j(t\mid a, x)}{S(u\mid a,x)S_C(u\mid a,x)} \dd M_1(u\mid a, x) \right. \nonumber \\
        & \phantom{\frac{\mathbb{1}(A=a, X=x)}{\pi(a\mid x)f(x)} } + \int_0^t \frac{F_j(u\mid a,x) - F_j(t\mid a, x)}{S(u\mid a,x)S_C(u\mid a,x)} \dd M_2(u\mid a, x) + \left. \int_0^t \frac{\dd M_j(u\mid a, x)}{S_C(u\mid a,x)} \right] \nonumber \\  
     = & \frac{\mathbb{1}(A=a, X=x)}{\pi(a\mid x)f(x)} \sum_{i=1,2} \int_0^t \frac{1}{S_C(u\mid a,x)} \left(\frac{F_j(u\mid a,x) - F_j(t\mid a, x)}{S(u\mid a,x)} + \mathbb{1}(i=j) \right) \dd M_i(u\mid a, x).  
    \end{align}
    Finally, by \eqref{gat:Fj}, we have
    \begin{align}
        & \partial_\epsilon L_{j, \epsilon}(0, t^*\mid a, x) \\
       = & \int_0^{t^*} \partial_\epsilon F_{j, \epsilon}(t\mid a, x) \dd t \nonumber \\
      = &  \int_0^{t^*} \sum_{i=1,2} \int_0^t \frac{1}{S_C(u\mid a,x)} \left(\frac{F_j(u\mid a,x) - F_j(t\mid a, x)}{S(u\mid a,x)} + \mathbb{1}(i=j) \right) \dd M_i(u\mid a, x) \dd t \nonumber \\
        & \quad \times \frac{\mathbb{1}(A=a, X=x)}{\pi(a\mid x)f(x)} \nonumber \\ 
      = & \int_0^{t^*} \sum_{i=1,2} \int_u^{t^*} \left(\frac{F_j(u\mid a,x) - F_j(t\mid a, x)}{S(u\mid a,x)} + \mathbb{1}(i=j) \right) \dd t\frac{1}{S_C(u\mid a,x)} \dd M_i(u\mid a, x) \nonumber \\
        & \times \frac{\mathbb{1}(A=a, X=x)}{\pi(a\mid x)f(x)} \nonumber \\
      = & \frac{\mathbb{1}(A=a, X=x)}{\pi(a\mid x)f(x)} \sum_{i=1,2} \int_0^{t^*} \frac{H_{ij}(s,t^*\mid a, x)}{S_C(s\mid a, x)} \dd M_i(s\mid a, x). \nonumber
    \end{align}
\end{proof}

\begin{proof}[\textbf{Proof of lemma \ref{lem:gatpsi}}]
    We have 
    $$
    \partial_\epsilon \psi_j(P_\epsilon) = \partial_\epsilon \E_{P_\epsilon}\{ \tau_{j,P_\epsilon}(X)\} = \tau_j(X) + \E\{\partial_\epsilon \tau_{j,P_\epsilon}(X) \} - \psi_j(P).
    $$
    Applying lemma \ref{lem:gatLj} gives the result.
\end{proof}

\begin{proof}[\textbf{Proof of lemma \ref{lem:gatvim}}]
    The EIF of $\chi^l(P)$ is given by theorem 3 in \textcite{zm}. By remark 2 in \textcite{zm}, the EIF of $\Gamma_j^l(P)$ is given by $\tilde{\psi}_{\Gamma_j^l}$ provided that the CATE function $\tau_j(x)$ has Gateaux derivative given by $\frac{\mathbb{1}(X=x)}{f(x)}\left[\varphi_j(O) - \tau_j(X) \right]$, which is seen to hold by lemma \ref{lem:gatLj}. The EIF of $\Omega_j^l$ then follows from the chain rule.
\end{proof}

\section{Proof of asymptotic results}\label{app:asympLin}
\begin{alemma}\label{lem:remate_a}
    The remainder term $P\{\varphi_j^a(\hat{\nu}) - \tau_j^a\}$, $a=0,1$, can be represented as
    \begin{align}
         &\E\left\{\sum_{i=1,2} \int_0^{t^*} S(s\mid a, X) \hat{H}_{ij}(s, t^*\mid a, X) \right. \nonumber \\
         & \left. \quad \quad \times \left(1 - \frac{\pi(a\mid X)S_C(s\mid a, X)}{\hat{\pi}(a\mid X)\hat{S}_C(s\mid a, X)} \right)\dd \left[\hat{\Lambda}_i(s\mid a, X) - \Lambda_i(s\mid a, X) \right] \right\} \nonumber
    \end{align}
    where the expectation is taken with respect to an observation $O$, considering the nuisance estimators, $\hat{\nu}$, fixed. Under assumption \ref{ass:nuisDouble} it holds that $P\{\varphi_j(\hat{\nu}) - \tau_j\} = o_p(n^{-1/2})$.
\end{alemma}
\begin{proof}[Proof of lemma \ref{lem:remate_a}]
    Throughout the proof, expectations and conditional expectations will be taken with respect to an observation $O$, considering estimated nuisance parameters fixed. For the first statement we have
    \begin{align} \label{rem:ate_a_expansion}
        P(\phi_j^a(\hat{\nu}) - \tau_j^a) =& \E\left\{ \hat{L}_j(0, t^*\mid a, X) - L_j(0, t^*\mid a, X)\right\} \nonumber \\
      & +  \E\left\{\frac{\mathbb{1}(A=a)}{\hat{\pi}(a\mid X)} \sum_{i = 1,2} \int_0^{t^*} \frac{\hat{H}_{ij}(s, t^*\mid a, X)}{\hat{S}_C(s\mid a, X)} \dd \hat{M}_i(s\mid a, X) \right\} 
    \end{align}
    and we will expand the two terms separately. For the first term note that
    \begin{align} \label{Fj-diff-tmp}
        \hat{F}_j(t \mid a, x) - F_j(t\mid a, x) =& \int_0^t \hat{S}(s\mid a, x) - S(s\mid a, x) \dd \hat{\Lambda}_j(s\mid a,x) \nonumber \\
        &+ \int_0^t S(s\mid a,x) \left[ \hat{\Lambda}_j(s\mid a,x) - \Lambda_j(s\mid a,x) \right].
    \end{align}
    Define 
    $$
    \Lambda(s\mid a,x) = \Lambda_1(s\mid a, x) + \Lambda_1(s\mid a, x),
    $$
    then Duhamel's equations (\cite{gill}) gives 
    $$
    \hat{S}(s\mid a, x) - S(s\mid a, x) = \int_0^s \frac{S(u\mid a, x)}{\hat{S}(u\mid a,x)} \dd \left[ \Lambda(u\mid a,x ) - \hat{\Lambda}(u\mid a,x ) \right]\hat{S}(s\mid a, x).
    $$
    Plugging this into the first term in \eqref{Fj-diff-tmp} gives
    \begin{align}
        & \int_0^t \hat{S}(s\mid a, x) - S(s\mid a, x) \dd \hat{\Lambda}_j(s\mid a,x) \nonumber \\
      = & \int_0^t \int_0^s \frac{S(u\mid a, x)}{\hat{S}(u\mid a,x)} \dd \left[ \Lambda(u\mid a,x ) - \hat{\Lambda}(u\mid a,x ) \right]\hat{S}(s\mid a, x) \dd \hat{\Lambda}_j(s\mid a,x) \nonumber \\
      = & \int_0^t \int_u^t \hat{S}(s\mid a, x) \frac{S(u\mid a, x)}{\hat{S}(u\mid a,x)} \dd \hat{\Lambda}_j(s\mid a,x) \dd \left[ \Lambda(u\mid a,x ) - \hat{\Lambda}(u\mid a,x ) \right] \nonumber \\
      = & \int_0^t \frac{S(u\mid a, x)}{\hat{S}(u\mid a,x)} \left(\hat{F}_j(t\mid a,x) - \hat{F}_j(u\mid a,x) \right) \dd \left[ \Lambda(u\mid a,x ) - \hat{\Lambda}(u\mid a,x ) \right] \nonumber \\
      = & \int_0^t \frac{S(u\mid a, x)}{\hat{S}(u\mid a,x)} \left(\hat{F}_j(t\mid a,x) - \hat{F}_j(u\mid a,x) \right) \dd \left[ \Lambda_1(u\mid a,x ) - \hat{\Lambda}_1(u\mid a,x ) \right] \nonumber \\
        & + \int_0^t \frac{S(u\mid a, x)}{\hat{S}(u\mid a,x)} \left(\hat{F}_j(t\mid a,x) - \hat{F}_j(u\mid a,x) \right) \dd \left[ \Lambda_2(u\mid a,x ) - \hat{\Lambda}_2(u\mid a,x ) \right] \nonumber 
    \end{align}
    Using this expansion, we can write \eqref{Fj-diff-tmp} as
    \begin{align} \label{Fj-diff-final}
        & \hat{F}_j(t \mid a, x) - F_j(t\mid a, x) \nonumber \\
      = & \int_0^t \frac{S(u\mid a, x)}{\hat{S}(u\mid a,x)} \left(\hat{F}_j(u\mid a,x) - \hat{F}_j(t\mid a,x) \right) \dd \left[ \hat{\Lambda}_1(u\mid a,x ) - \Lambda_1(u\mid a,x) \right] \nonumber \\
        & + \int_0^t \frac{S(u\mid a, x)}{\hat{S}(u\mid a,x)} \left(\hat{F}_j(u\mid a,x) - \hat{F}_j(t\mid a,x) \right) \dd \left[ \hat{\Lambda}_2(u\mid a,x ) - \Lambda_2(u\mid a,x) \right] \nonumber \\
        & + \int_0^t S(s\mid a,x) \dd \left[ \hat{\Lambda}_j(s\mid a,x) - \Lambda_j(s\mid a,x) \right] \nonumber \\
      = & \sum_{i=1,2}\int_0^t S(s\mid a,x) \left(\mathbb{1}(i = j) + \frac{\hat{F}_j(s\mid a,x) - \hat{F}_j(t\mid a,x)}{\hat{S}(s\mid a, x)} \right) \dd \left[ \hat{\Lambda}_i(s\mid a,x ) - \Lambda_i(s\mid a,x) \right] 
    \end{align}
    For the second term in \eqref{rem:ate_a_expansion} note that 
    \begin{align}\label{rem:ate_a-ite-expec}
        & \E\left( \int_0^{t^*} \frac{\hat{H}_{ij}(s, t^*\mid A, X)}{\hat{S}_C(s\mid A, X)} \dd \hat{M}_i(s\mid A, X) \ \biggl| \ A, X \right) \nonumber \\
      = &  \int_0^{t^*} \frac{\hat{H}_{ij}(s, t^*\mid A, X)}{\hat{S}_C(s\mid A, X)} \E(\dd N_i(s) \mid A, X) \nonumber \\ 
        & - \int_0^{t^*} \frac{\hat{H}_{ij}(s, t^*\mid A, X)}{\hat{S}_C(s\mid A, X)} \E(\mathbb{1}(\tilde{T}\geq s) \mid A, X) \dd \hat{\Lambda}_i(s\mid A, X) \nonumber \\
      = & \int_0^{t^*} \frac{\hat{H}_{ij}(s, t^*\mid A, X)}{\hat{S}_C(s\mid A, X)} S(s\mid A, X) S_C(s\mid A, X) \dd \left[ \Lambda_i(s\mid A, X) - \hat{\Lambda}_i(s\mid A, X) \right].
    \end{align}
    Hence, by collecting \eqref{Fj-diff-final} and \eqref{rem:ate_a-ite-expec}, and using iterated expectation, we can write \eqref{rem:ate_a_expansion} as
    \begin{align}
         & \E\left\{ \sum_{i=1,2}\int_0^{t^*}\int_0^u S(s\mid a,X) \left(\mathbb{1}(i = j) + \frac{\hat{F}_j(s\mid a,X) - \hat{F}_j(u\mid a,X)}{\hat{S}(s\mid a, X)} \right) \dd \left[ \hat{\Lambda}_i(s\mid a,X ) - 
           \Lambda_i(s\mid a,X) \right] \dd u \right. \nonumber \\
         &  \left. -\frac{\pi(a\mid X)}{\hat{\pi}(a\mid X)} \sum_{i=1,2} \int_0^{t^*} \frac{\hat{H}_{ij}(s, t^*\mid a, X)}{\hat{S}_C(s\mid a, X)} S(s\mid a, X) S_C(s\mid a, X) \dd 
           \left[ \hat{\Lambda}_i(s\mid a, X) - \Lambda_i(s\mid a, X) \right] \right\} \nonumber \\
       = & \E\left\{ \sum_{i=1,2} \int_0^{t^*} S(s\mid a, X) \hat{H}_{ij}(s,t^*\mid a, X) \right. \\
         & \left. \quad \quad \times \left( 1 - \frac{\pi(a\mid X) S_C(s\mid a, X)}{\hat{\pi}(a\mid X) \hat{S}_C(s\mid a, X)} \right) \dd 
           \left[ \hat{\Lambda}_i(s\mid a, X) - \Lambda_i(s\mid a, X) \right] \right\} \nonumber    
    \end{align}
    which gives the first statement. For the second statement note that 
    $$
    P\{\phi^j(\hat{\nu}) - \tau^j\} = P\{\phi^j_1(\hat{\nu}) - \tau^j_1\} - P\{\phi^j_0(\hat{\nu}) - \tau^j_0\}.
    $$
    Applying the representation above along with assumption \ref{ass:nuisDouble} gives the result.
\end{proof}

\begin{alemma}\label{lem:empate}
    Let $g(s\mid a,x) = \pi(a\mid x) S_C(s\mid a,x)$ and $\hat{g}(s\mid a,x) = \hat{\pi}(a\mid x) \hat{S}_C(s\mid a,x)$. Under assumption \ref{ass:nuisPos} and \ref{ass:nuisCons}, the uncentered EIF, $\varphi_j$, is bounded in $L_2(P)$-norm as: 
    \begin{align*}
         \norm{\varphi_j(\hat{\nu}) - \varphi_j(\nu)} \leq & C_1^*\norm{ \sup_{s\leq t^*}\left| \hat{\Lambda}_1(s\mid a, X) - \Lambda_1(s\mid a, X) \right|} \\
         & + C_2^*\norm{ \sup_{s\leq t^*}\left| \hat{\Lambda}_2(s\mid a, X) - \Lambda_2(s\mid a, X) \right|} \\
         & + C_c^*\norm{ \sup_{s\leq t^*}\left| \hat{g}(s\mid a, X) - g(s\mid a, X) \right|}
    \end{align*}
\end{alemma}

\begin{proof}[Proof of lemma \ref{lem:empate}]
    Note that 
    $$
    \norm{\varphi_j(\hat{\nu}) - \varphi_j(\nu)} \leq \norm{\varphi_j^1(\hat{\nu}) - \varphi_j^1(\nu)} + \norm{\varphi_j^0(\hat{\nu}) - \varphi_j^0(\nu)}.
    $$
    Hence we need to show $\norm{\varphi_j^a(\hat{\nu}) - \varphi_j^a(\nu)} = o_p(1)$ for $a=0,1$. Observe that 
    \begin{align}
        & P\{\varphi_j^a(\hat{\nu}) - \varphi_j^a(\nu)\}^2 \nonumber \\
   \leq & 2\E\left\{ \hat{L}_j(0,t^*\mid a, X) - L_j(0,t^*\mid a, X)  \right\}^2 \label{emp:1} \\
        & + 2\E\left\{ \sum_{i=1,2}\int_0^{t^*} \frac{\hat{H}_{ij}(s,t^*\mid a, X)}{\hat{g}(s\mid a, X)} - \frac{H_{ij}(s,t^*\mid a, X)}{g(s\mid a, X)} \dd N_i(s) \right\}^2 \label{emp:2} \\
        & + 2\E\left\{ \sum_{i=1,2}\left[\int_0^{t^*\wedge \tilde{T}} \frac{\hat{H}_{ij}(s,t^*\mid a, X)}{\hat{g}(s\mid a, X)} \dd \hat{\Lambda}_i(s\mid a,X) - \int_0^{t^*\wedge \tilde{T}} \frac{H_{ij}(s,t^*\mid a, X)}{g(s\mid a, X)} \dd \Lambda_i(s\mid a, X) \right] \right\}^2. \label{emp:3}
    \end{align}
    We will deal with each of the terms separately, but first we will derive some results for the consistency of $\hat{F}_j$, $\hat{L}_j$ and $\hat{H}_{ij}$. For $\hat{F}_j$ we have that 
    \begin{align} \label{bound:Fj}
        &\hat{F}_j(t^*\mid a, x) - F_j(t^*\mid a, x) \nonumber \\
      = & \int_0^{t^*} \hat{S}(s\mid a, x) \dd \left[\hat{\Lambda}_j(s\mid a, x) - \Lambda_j(s\mid a, x) \right] + \int_0^{t^*} \hat{S}(s\mid a,x) - S(s\mid a,x) \dd \Lambda_j(s\mid a, x) \nonumber \\
      = & \hat{S}(t^*\mid a,x)\left[ \hat{\Lambda}_j(t^* \mid a,x) - \Lambda_j(t^* \mid a,x) \right] - \hat{S}(0\mid a,x)\left[ \hat{\Lambda}_j(0 \mid a,x) - \Lambda_j(0 \mid a,x) \right] \nonumber \\
        & - \int_0^{t^*} \hat{\Lambda}_j(s\mid a, x) - \Lambda_j(s\mid a, x) \dd \hat{S}(s\mid a, x)  + \int_0^{t^*} \hat{S}(s\mid a,x) - S(s\mid a,x) \dd \Lambda_j(s\mid a, x)  \nonumber \\
   \leq & 2\sup_{s\leq t^*}\left| \hat{\Lambda}_j(s\mid a, x) - \Lambda_j(s\mid a, x) \right| + \hat{\Lambda}_j(t\mid a,x)\sup_{s\leq t^*}\left| \hat{S}(s\mid a, x) - S(s\mid a, x) \right|  \nonumber \\
   \leq & 2\sup_{s\leq t^*}\left| \hat{\Lambda}_j(s\mid a, x) - \Lambda_j(s\mid a, x) \right| \nonumber \\
        & + \log(\eta^{-1}) e^C \sup_{s\leq t^*} \left[  \left| \hat{\Lambda}_1(s\mid a, x) - \Lambda_1(s\mid a, x) \right| + \left| \hat{\Lambda}_2(s\mid a, x) - \Lambda_2(s\mid a, x) \right| \right] \nonumber \\
      = & C_1 \sup_{s\leq t^*}\left| \hat{\Lambda}_1(s\mid a, x) - \Lambda_1(s\mid a, x) \right| + C_2 \sup_{s\leq t^*}\left| \hat{\Lambda}_2(s\mid a, x) - \Lambda_2(s\mid a, x) \right|
    \end{align}
    for constants $C_1>0$ and $C_2 > 0$. The second equality follows from partial integration and the second inequality follows from the mean value theorem with $C>0$ is some value between $\hat{\Lambda}(s\mid a,x)$ and $\Lambda(s\mid a,x)$ together with assumption \ref{ass:nuisPos}.

    From \eqref{bound:Fj} it follows immediately that 
    \begin{align}\label{bound:Lj}
        &\hat{L}_j(0,t^*\mid a, x) - L_j(0,t^*\mid a, x) \nonumber \\
   \leq & t^*C_1 \sup_{s\leq t^*}\left| \hat{\Lambda}_1(s\mid a, x) - \Lambda_1(s\mid a, x) \right| + t^*C_2 \sup_{s\leq t^*}\left| \hat{\Lambda}_2(s\mid a, x) - \Lambda_2(s\mid a, x) \right|.
    \end{align}
    For $\hat{H}_{ij}$, observe
    \begin{align} \label{bound:Hij}
        & \hat{H}_{ij}(s,t^*\mid a,x) - H_{ij}(s,t^*\mid a,x) \nonumber \\
      = & \int_s^{t^*} \frac{\hat{F}_j(s\mid a,x) - F_j(s\mid, a, x) - (\hat{F}_j(u\mid a,x) - F_j(u\mid, a, x))}{\hat{S}(s\mid a,x)} \dd u \nonumber \\
        & + \int_s^{t^*}\left(\frac{1}{\hat{S}(s\mid a,x)} - \frac{1}{S(s\mid a,x)} \right) \left( F_j(s\mid a,x) - F_j(u\mid a, x)  \right) \dd u \nonumber \\
   \leq & t^* \left|\frac{\hat{F}_j(s\mid a,x) - F_j(s\mid a,x)}{\hat{S}(s\mid a,x)} \right| + \eta^{-1}\int_s^{t^*}\left| \hat{F}_j(u\mid a,x) - F_j(u\mid a,x) \right| \dd u \nonumber \\
        & + \eta^{-2} \left| \hat{S}(s\mid a,x) - S(s\mid a,x) \right| \left|F_j(s\mid a,x)(t^* - s) - \int_s^{t^*} F_j(u\mid a,x) \dd u \right| \nonumber \\
   \leq &   C_3 \sup_{s\leq t^*}\left| \hat{\Lambda}_1(s\mid a, x) - \Lambda_1(s\mid a, x) \right| + C_4 \sup_{s\leq t^*}\left| \hat{\Lambda}_2(s\mid a, x) - \Lambda_2(s\mid a, x) \right|     
    \end{align}
    for some constants $C_3>0$ and $C_4 > 0$, which follows from \eqref{bound:Fj}. Now we proceed to bound each of the three terms initially stated. The bound of \eqref{emp:1} follows directly from \eqref{bound:Lj}:
    \begin{align}\label{bound:emp1}
        & 2\E\left\{ \hat{L}_j(0,t^*\mid a, X) - L_j(0,t^*\mid a, X)  \right\}^2 \nonumber \\
        \leq &2\E\left\{ C_5 \sup_{s\leq t^*}\left| \hat{\Lambda}_1(s\mid a, x) - \Lambda_1(s\mid a, x) \right| + C_6 \sup_{s\leq t^*}\left| \hat{\Lambda}_2(s\mid a, x) - \Lambda_2(s\mid a, x) \right|  \right\}^2.
    \end{align}
    For \eqref{emp:2} we have
    \begin{align} \label{bound:emp2}
        &\E\left\{ \sum_{i=1,2}\int_0^{t^*} \frac{\hat{H}_{ij}(s,t^*\mid a, X)}{\hat{g}(s\mid a, X)} - \frac{H_{ij}(s,t^*\mid a, X)}{g(s\mid a, X)} \dd N_i(s) \right\}^2 \nonumber \\
      = & \E\left\{ \sum_{i=1,2} \int_0^{t^*} \frac{\hat{H}_{ij}(s,t^*\mid a, X) - H_{ij}(s,t^*\mid a, X)}{\hat{g}(s\mid a, X)} \dd N_i(s)   \right. \nonumber \\
        & + \left. \int_0^{t^*} H_{ij}(s,t^*\mid a,x) \left(\frac{1}{\hat{g}(s\mid a,X)} - \frac{1}{g(s\mid a,X)} \right) \dd N_i(s)   \right\}^2 \nonumber \\
   \leq & \E\left\{ \sum_{i=1,2} \eta^{-2} \sup_{s\leq t^*}\left| \hat{H}_{ij}(s\mid a,X) - H_{ij}(s\mid a,X) \right| \right. \\
        & \left. + \eta^{-4}\sup_{s\leq t^*} \left|\hat{g}(s\mid a,X) - g(s\mid a,X) \right| \left|H_{ij}(s,t^*\mid a,X) \right| \right\}^2 \nonumber \\
   \leq & \E \left\{ C_7 \sup_{s\leq t^*}\left| \hat{\Lambda}_1(s\mid a, X) - \Lambda_1(s\mid a, X) \right| + C_8 \sup_{s\leq t^*}\left| \hat{\Lambda}_2(s\mid a, X) - \Lambda_2(s\mid a, X) \right| \right. \nonumber \\
        & \left. \phantom{\E\{ \left| \hat{\Lambda} \right| } + C_9 \sup_{s\leq t^*} \left|\hat{g}(s\mid a,X) - g(s\mid a,X) \right|\right\}^2.
    \end{align}
    The first inequality follows from assumption \ref{ass:nuisPos} and the second inequality follows from \eqref{bound:Hij} together with 
    \begin{align}\label{bound:Hijtmp}
        H_{ij}(s,t^*\mid a,x) \leq & \eta^{-1}\int_s^{t^*}\left| F_j(s) - F_j(u) \right| \dd u + t^* - s \nonumber \\
                              \leq & \eta^{-1} 2 t^* 
    \end{align}
    by assumption \ref{ass:nuisPos}. 

    Lastly, for \eqref{emp:3} we have 
    \begin{align}\label{emp:3tmp}
        &\E\left\{ \sum_{i=1,2}\left[\int_0^{t^*\wedge \tilde{T}} \frac{\hat{H}_{ij}(s,t^*\mid a, X)}{\hat{g}(s\mid a, X)} \dd \hat{\Lambda}_i(s\mid a,X) - \int_0^{t^*\wedge \tilde{T}} \frac{H_{ij}(s,t^*\mid a, X)}{g(s\mid a, X)} \dd \Lambda_i(s\mid a, X) \right] \right\}^2 \nonumber \\
      = & E\left\{ \underbrace{\sum_{i=1,2} \left[ \int_0^{t^*\wedge \tilde{T}} \hat{H}_{ij}(s,t^*\mid a, X) \left( \frac{1}{\hat{g}(s\mid a,x)} - \frac{1}{g(s\mid a,x)} \right) \dd \hat{\Lambda}_i(s\mid a,x)  \right]}_{(a)} \right. \nonumber \\
        & \left. + \underbrace{ \int_0^{t^*\wedge \tilde{T}} \sum_{i=1,2} \frac{\hat{H}_{ij}(s,t^*\mid a, X)}{g(s\mid a,X)} \dd \hat{\Lambda}_i(s\mid a,X) - \int_0^{t^*\wedge \tilde{T}} \sum_{i=1,2} \frac{H_{ij}(s\mid a, X)}{g(s\mid a,X)} \dd \Lambda_i(s\mid a,X)}_{(b)} \right\}^2.
    \end{align}
    and we will bound $(a)$ and $(b)$ separately. For $(a)$ we have that for almost all $x$
    \begin{align} \label{bound:a}
        & \int_0^{t^*\wedge \tilde{T}} \hat{H}_{ij}(s, t^*\mid a, x) \left( \frac{1}{\hat{g}(s\mid a,x)} - \frac{1}{g(s\mid a,x)} \right) \dd \hat{\Lambda}_i(s\mid a,x) \nonumber \\
   \leq & \sum_{i=1,2} \sup_{s\leq t^*}\left|\hat{H}_{ij}(s, t^* \mid a, x)\right| \sup_{s\leq t^*}\left| \frac{1}{\hat{g}(s\mid a,x)} - \frac{1}{g(s\mid a,x)} \right| \hat{\Lambda}_i(t^*\mid a,x) \nonumber \\
   \leq & 4\eta^{-5}t^*\log(\eta) \sup_{s\leq t^*}\left| \hat{g}(s\mid a,x) - g(s\mid a,x) \right|.
    \end{align}
    by \eqref{bound:Hijtmp} together with assumption \ref{ass:nuisPos}. Next, to control (b), define 
    $$
    H_j(s, t^* \mid a,x) = \int_s^{t^*} \frac{F_j(s\mid a,x) - F_j(u\mid a,x)}{S(s)} \dd u
    $$
    and let 
    $$\Lambda(s\mid a,x) = \Lambda_1(s\mid a,x) + \Lambda_2(s\mid a,x)$$
    such that $S(s\mid a,x) = e^{-\Lambda(s\mid a,x)}$. Observe that 
    \begin{align*}
        & H_j(s,t^*\mid a,x) \\
      = & - \int_s^{t^*} \frac{F_j(u\mid a,x) - F_j(s\mid a,x)}{S(s)} \dd u \\
      = & -\int_s^{t^*} \int_s^u \frac{S(v\mid a,x)}{S(s\mid a,x)} \Lambda_j(dv\mid a,x) \dd u \\
      \overset{(*)}{=} & - \int_s^{t^*} \int_s^u \left(1 - \int_s^v \frac{S(v\mid a,x)}{S(w\mid a,x)} \Lambda(dw\mid a,x) \right) \Lambda_j(dv\mid a,x) \dd u \\
      = & - \left[ \int_s^{t^*} \int_s^u  \Lambda_j(\dd v\mid a,x) \dd u - \int_s^{t^*} \int_s^u \int_s^v \frac{S(v\mid a,x)}{S(w\mid a,x)} \Lambda(\dd w\mid a,x) \Lambda_j(\dd v\mid a,x) \dd u  \right]\\
      = & - \left[ \int_s^{t^*} \int_v^{t^*} \dd u \Lambda_j(\dd v\mid a,x) - \int_s^{t^*} \int_w^{t^*} \int_w^u \frac{S(v\mid a,x)}{S(w\mid a,x)}  \Lambda_j(dv\mid a,x) \dd u \Lambda(dw\mid a,x) \right] \\
      = & - \left[ \int_s^{t^*} \int_v^{t^*} \dd u \Lambda_j(\dd v\mid a,x) - \int_s^{t^*} \int_w^{t^*} \frac{F_j(u\mid a,x) - F_j(w\mid a,x)}{S(w\mid a,x)} \dd u \Lambda(dw\mid a,x) \right] \\
      = & - \int_s^{t^*} \sum_{i=1,2} H_{ij}(w,t^*\mid a,x) \Lambda_i(\dd w\mid a,x)
    \end{align*}
    where $(*)$ follows from the backward equation (theorem 5, \cite{gill}). Hence 
    $$H_j(\dd s,t^*\mid a,x) = \sum_{i=1,2} H_{ij}(s,t^*\mid a,x) \Lambda_i(\dd s\mid a,x).$$
    From this expression we can derive a bound for $(b)$ using integration by parts:  
    \begin{align} \label{bound:b}
         & \int_0^{t^*\wedge \tilde{T}} \sum_{i=1,2} \frac{\hat{H}_{ij}(s,t^*\mid a, x)}{g(s\mid a,x)} \dd \hat{\Lambda}_i(s\mid a,x) - \int_0^{t^*\wedge \tilde{T}} \sum_{i=1,2} \frac{H_{ij}(s\mid a, x)}{g(s\mid a,x)} \dd \Lambda_i(s\mid a,x) \nonumber\\
       = & \int_0^{t^*\wedge \tilde{T}} \frac{1}{g(s\mid a,x)} \dd \left[\hat{H}_j(s, t^*\mid a,x) - H_j(s, t^*\mid a,x) \right] \nonumber \\
       = & \frac{1}{g(t^*\mid a,x)} \left[\hat{H}_j(t^*\wedge \tilde{T}, t^*\mid a,x) - H_j(t^*\wedge \tilde{T}, t^*\mid a,x) \right] \nonumber \\
         & - \frac{1}{g(0\mid a,x)}\left[\hat{H}_j(0, t^*\mid a,x) - H_j(0, t^*\mid a,x) \right] \nonumber\\
         & - \int_0^{t^*\wedge \tilde{T}} \left[\hat{H}_j(s, t^*\mid a,x) - H_j(s, t^*\mid a,x) \right] \left(\frac{1}{g} \right)(ds\mid a,x) \nonumber \\
    \leq & 3\eta^{-1}\sup_{s\leq t^*} \left|\hat{H}_j(s, t^*\mid a,x) - H_j(s, t^*\mid a,x) \right| \nonumber \\
    \leq & C_{10} \sup_{s\leq t^*}\left| \hat{\Lambda}_1(s\mid a, x) - \Lambda_1(s\mid a, x) \right| + C_{11} \sup_{s\leq t^*}\left| \hat{\Lambda}_2(s\mid a, x) - \Lambda_2(s\mid a, x) \right|
    \end{align}
    for some constants $C_{10}>0$ and $C_{11}>0$. Applying the bounds \eqref{bound:a} and \eqref{bound:b} to \eqref{emp:3tmp} gives that \eqref{emp:3} is bounded by
    \begin{align}\label{bound:emp3}
        & E\left\{ C_{10} \sup_{s\leq t^*}\left| \hat{\Lambda}_1(s\mid a, X) - \Lambda_1(s\mid a, X) \right| + C_{11} \sup_{s\leq t^*}\left| \hat{\Lambda}_2(s\mid a, X) - \Lambda_2(s\mid a, X) \right| \right. \nonumber \\
        & + \left. C_{12} \sup_{s\leq t^*}\left| \hat{g}_c(s\mid a, X) - g(s\mid a, X) \right|  \right\}^2.
    \end{align}
    Thus, applying \eqref{bound:emp1}, \eqref{bound:emp2} and \eqref{bound:emp3} to \eqref{emp:1}, \eqref{emp:2} and \eqref{emp:3}, respectively, gives the result.

\end{proof}

\begin{proof}[\textbf{Proof of Theorem \ref{thm:anATE}}]
    Consider the decomposition
    \begin{align*}
        \mathbb{P}_n^k \varphi(\hat{\nu}_{-k}) = \mathbb{P}_n^k\tilde{\psi}_{\psi_j} + (\mathbb{P}_n^k - P)(\varphi(\hat{\nu}_{-k}) - \varphi(\nu)) + P\varphi(\hat{\nu}_{-k})
    \end{align*}
    such that 
    \begin{align*}
        \hat{\psi}_j^{CF} - \psi_j &= \sum_{k=1}^K \frac{n_k}{n}\mathbb{P}_n^k\varphi(\hat{\nu}_{-k}) - \psi_j \\
        &= \mathbb{P}_n\tilde{\psi}_{\psi_j} + \underbrace{\sum_{k=1}^K \frac{n_k}{n}(\mathbb{P}_n^k - P)(\varphi(\hat{\nu}_{-k}) - \varphi(\nu))}_{\text{empirical process term}} + \underbrace{\sum_{k=1}^K \frac{n_k}{n} P(\varphi(\hat{\nu}_{-k}) - \psi_j)}_{\text{remainder term}}.
    \end{align*}
    Now, if both the empirical process term and the remainder term in the above display are $o_p(n^{-1/2})$, it follows that $\hat{\psi}_j^{CF}$ is asymptotically linear with influence function given by $\tilde{\psi}_j$. By Proposition 2 in \textcite{kennedydouble} this is achieved if $\norm{\varphi(\hat{\nu}_{-k}) - \varphi(\nu)}=o_p(1)$ for each $k$ and if the remainder is $o_p(n^{-1/2})$. Under assumption \ref{ass:nuis} for each $k$, the former is achieved by lemma \ref{lem:empate} and the latter is achieved by lemma \ref{lem:remate_a}. An application of the central limit theorem together with Slutsky's lemma gives the convergence in distribution. 
\end{proof}

\begin{proof}[\textbf{Proof of Theorem \ref{thm:anbplp}}]
    We will show that $\hat{\Gamma}_j^{l,CF}$ and $\hat{\chi}_j^{l,CF}$ are asymptotically linear with influence function given by $\tilde{\psi}_{\Gamma_j^l}$ and $\tilde{\psi}_{\chi_j^l}$, respectively. Since $\hat{\Omega}_j^{l,CF}$ is a ratio of the two, it follows from the functional delta method that it is asymptotically linear with influence function given by $\tilde{\psi}_{\Omega_j}$ (\cite{Vaart}, Ch. 25.7). 
    
    That $\hat{\chi}^{l,CF}$ is asymptotically linear follows from Theorem 5 in \textcite{zm}. To show that $\hat{\Gamma}_j^{l,CF}$ is asymptotically linear, we consider the decomposition 
    \begin{align*}
        \mathbb{P}_n^k \phi_{\Gamma_j^l}(\hat{\nu}_{-k}) = \mathbb{P}_n^k\tilde{\psi}_{\Gamma_j^l} + (\mathbb{P}_n^k - P)(\phi_{\Gamma_j^l}(\hat{\nu}_{-k}) - \phi_{\Gamma_j^l}(\nu)) + P\phi_{\Gamma_j^l}(\hat{\nu}_{-k})
    \end{align*}
    such that 
    \begin{align*}
        \hat{\Gamma}_j^{l,CF} - \psi_j &= \sum_{k=1}^K \frac{n_k}{n}\mathbb{P}_n^k\phi_{\Gamma_j^l}(\hat{\nu}_{-k}) - \Gamma_j^l \\
        &= \mathbb{P}_n\tilde{\psi}_{\Gamma_j^l} + \underbrace{\sum_{k=1}^K \frac{n_k}{n}(\mathbb{P}_n^k - P)(\phi_{\Gamma_j^l}(\hat{\nu}_{-k}) - \phi_{\Gamma_j^l}(\nu))}_{\text{empirical process term}} + \underbrace{\sum_{k=1}^K \frac{n_k}{n} P(\phi_{\Gamma_j^l}(\hat{\nu}_{-k}) - \Gamma_j^l)}_{\text{remainder term}}.
    \end{align*}
    Again, by Proposition 2 in \textcite{kennedydouble}, the desired asymptotic linearity follows if we can show that $\norm{\phi_{\Gamma_j^l}(\hat{\nu}_{-k}) - \phi_{\Gamma_j^l}(\nu)} = o_p(1)$ for each $k$ and that the remainder term is $o_p(n^{-1/2})$. For both results, we will consider arguments that are similar to the ones given in the proof of Theorem 4 in \textcite{zm}. \\ \\
    \textit{Empirical process term} \\ \\
    Consider the following expansion for a given $k$:
    \begin{align*}
        \phi_{\Gamma_j^l}(\hat{\nu}_{-k})(O) - \phi_{\Gamma_j^l}(\nu)(O) 
      =& [\varphi_j(\hat{\nu}_{-k})(O) - \varphi_j(\nu)(O)][X_l - \hat{E}^l_{-k}(X_{-l})] \\
    &- [\hat{\tau}_{j,-k}^l(X_{-l}) - \tau_j^l(X_{-l})][X_l - \hat{E}_{-k}^l(X_{-l})] \\
    &- [\hat{E}_{-k}^l(X_{-l}) - E(X_l\mid X_{-l})][\varphi_j(\nu)(O) - \tau_j^l(X_{-l})].  
    \end{align*}
    For the first term we have
    $$
    \E\left\{[\varphi_j(\hat{\nu}_{-k})(O) - \varphi_j(\nu)(O)][X_l - \hat{E}^l_{-k}(X_{-l})]\right\}^2 \leq M \E\left\{\varphi_j(\hat{\nu}_{-k})(O) - \varphi_j(\nu)(O)\right\}^2 = o_p(1),
    $$
    where the inequality follows from assumption (i) and the equality follows from lemma \ref{lem:empate}, since we have assumed that assumption \ref{ass:nuisCons} holds for each $k$. Consistency in $L_2(P)$ of the second term follows by analogous arguments, replacing assumption \ref{ass:nuisCons} with assumption (ii). Consistency of the third term in $L_2(P)$ follows from assumption (iii) if $[\varphi_j(\nu)(O) - \tau_j^l(X_{-l})]^2$ is bounded almost surely. To that end, observe that for almost all $o \in \mathcal{O}$ we have 
    \begin{align*}
        &[\varphi_j(\nu)(o) - \tau_j^l(x_{-l})]^2 \\
        \leq & 2\left[\tau(x) - \tau_j^l(x_{-l})\right]^2 + 4 \left[ \sum_{i=1,2}\frac{\mathbb{1}(A=1)}{\pi(1\mid x)}\int_0^{t^*} \frac{H_{ij}(s,t^*\mid 1, x)}{S_c(s\mid 1,x)}\dd M_i(s\mid,a,x) \right]^2 \\
        & + 4\left[ \sum_{i=1,2} \frac{\mathbb{1}(A=0)}{\pi(0\mid x)}\int_0^{t^*} \frac{H_{ij}(s,t^*\mid 0, x)}{S_c(s\mid 0,x)}\dd M_i(s\mid,0,x) \right]^2.
    \end{align*}
    Hence consistency of the empirical process term follows, if we can bound each term in the display above. The first term is clearly bounded, since $\tau$ is bounded by $t^*$. For the second and third term observe that 
    \begin{align*}
        &\left[ \sum_{i=1,2} \frac{\mathbb{1}(A=a)}{\pi(a\mid x)}\int_0^{t^*} \frac{H_{ij}(s,t^*\mid a, x)}{S_c(s\mid a,x)}\dd M_i(s\mid,a,x) \right]^2 \\
      \leq & 2\eta^{-2}\left[ \sum_{i=1,2} \int_0^{t^*} H_{ij}(s,t^*\mid a, x)\dd N_i(s\mid,a,x) \right]^2 \\
           & + 2\eta^{-2}\left[ \sum_{i=1,2} \int_0^{t^*} H_{ij}(s,t^*\mid a, x)\mathbb{1}(\tilde{T}\geq s) \dd \Lambda_i(s\mid,a,x) \right]^2 \\
      \leq & 2 \eta^{-2}\left[ \sum_{i=1,2} 2\eta^{-1}t^* \right]^2 + 2\eta^{-2}\left[  2\eta^{-1}t^* \sum_{i=1,2} \int_0^{t^*} \mathbb{1}(\tilde{T}\geq s) \dd \Lambda_i(s\mid,a,x) \right]^2 \\
         = & 32 \eta^{-4}(t^*)^2 + 16\eta^{-4} (t^*)^2 \left[ \int_0^{t^*} \mathbb{1}(\tilde{T}\geq s) \dd \Lambda(s\mid,a,x) \right]^2 \\
      \leq & \eta^{-4}(t^*)^2(32 + 16 \log(\eta^{-1})^2)
    \end{align*}
    where the first and third inequality follows from assumption \ref{ass:nuisPos} and the second inequality follows from \eqref{bound:Hijtmp}. Thus it follows that $\norm{\phi_{\Gamma_j^l}(\hat{\nu}_{-k}) - \phi_{\Gamma_j^l}(\nu)} = o_p(1)$ for each $k$. \\ \\
    \textit{Remainder term} \\ \\
    As in \textcite{zm}, we consider the decomposition
    \begin{align} \label{eq:remgam}
    P\{ \phi_{\Gamma_j^l}(\hat{\nu}_{-k}) - \phi_{\Gamma_j}(\nu) \} = & E\left\{ [\varphi(\hat{\nu}_{-k})(O) - \varphi_j(\nu)(O)][X_l - \hat{E}_{-k}^l(X_{-l})] \right\} \nonumber \\
    &- \E\left\{[\hat{\tau}_{j,-k}^l(X_{-l}) - \tau_j^l(X_{-l})][X_j - \hat{E}_{-k}^j(X_{-l})]\right\} \nonumber \\
    &-\E\left\{ [\hat{E}_{-k}^l(X_{-j}) - E(X_l\mid X_{-l})][\varphi_j(\nu)(O) - \tau_j^l(X_{-l})]  \right\}
\end{align}
and we want to show that each term is $o_p(n^{-1/2})$. For the third term we note that 
\begin{align*}   
&E(\varphi_j(\nu)(O) - \tau_j^l(X_{-l})\mid A,X) \\
=& \left(\frac{\mathbb{1}(A=1)}{\pi(1\mid X)} - \frac{\mathbb{1}(A=0)}{\pi(0\mid X)}\right)\sum_{i=1,2} E\left( \int_0^{t^*} \frac{H_{ij}(s,t^*\mid A, X)}{S_C(s\mid A, X)} \dd M_i(s\mid A, X) \bigg| A, X \right) \\
=& 0
\end{align*}
since the integral is a martingale conditional on $A$ and $X$, and hence, that the third term is equal to 0 by iterated expectation. The second term is $o_p(n^{-1/2})$ by Cauchy-Schwarz together with assumption (ii) and (iii). Following the derivations in the proof of \ref{lem:remate_a}, we can use iterated expectation to write the first term as
\begin{align*}
    &\E\left\{[X_l - \hat{E}_{-k}^l(X_{-l})]\sum_{a=0,1}\sum_{i=1,2} \int_0^{t^*} S(s\mid a, X) \hat{H}_{ij}(s, t^*\mid a, X) \right. \\
    & \left. \phantom{\{[X_l - \hat{E}_{-k}^l(X_{-l})]\sum_{a=0,1}\sum_{i=1,2}} \times \left(1 - \frac{\pi(a\mid X)S_C(s\mid a, X)}{\hat{\pi}(a\mid X)\hat{S}_C(s\mid a, X)} \right)\dd \left[ 
        \hat{\Lambda}_i(s\mid a, X) - \Lambda_i(s\mid a, X) \right] \right\}.
\end{align*}
By assumption (i) this is bounded by
\begin{align*}
&\sqrt{M}\sum_{a=0,1}\E\left\{\left| \sum_{i=1,2} \int_0^{t^*} S(s\mid a, X) \hat{H}_{ij}(s, t^*\mid a, X) \right. \right. \\
 & \left. \left. \phantom{S(s\mid a, X) \hat{H}_{ij}(s, t^*\mid a, X) } \times \left(1 - \frac{\pi(a\mid X)S_C(s\mid a, X)}{\hat{\pi}(a\mid X)\hat{S}_C(s\mid a, X)} \right)\dd \left[ 
        \hat{\Lambda}_i(s\mid a, X) - \Lambda_i(s\mid a, X) \right] \right| \right\}
\end{align*}
which is $o_p(n^{-1/2})$ by assumption \ref{ass:nuisDouble}.

We have now shown that summands of the empirical process term and remainder term are $o_p(n^{-1/2})$ for each $k$, and hence it follows from Proposition 2 in \textcite{kennedydouble} that $\hat{\Gamma}_j^{l,CF}$ is asymptotically normal with influence function given by $\tilde{\psi}_{\Gamma_j^l}$. The convergence in distribution of $\sqrt{n}(\hat{\Omega}_j^{l,CF} - \Omega_j^l)$ follows from the central limit theorem together with Slutsky's lemma. 
\end{proof}
\end{document}